\documentclass[11pt,a4paper]{article}
\usepackage[utf8]{inputenc}
\usepackage[T1]{fontenc}
\usepackage{lmodern} 
\usepackage[margin=1in]{geometry}

\usepackage{amsmath,amsfonts,amssymb,amsthm,bm}

\usepackage[english]{babel}
\usepackage{bbm}
\usepackage{algpseudocode}
\usepackage{graphicx}
\usepackage{subcaption}
\usepackage{float} 
\usepackage{microtype}

\usepackage[colorlinks=true,linkcolor=blue,urlcolor=blue,citecolor=blue]{hyperref}

\usepackage[margin=1in]{geometry}

\usepackage[ruled,vlined,linesnumbered]{algorithm2e}

\usepackage{xcolor}

\theoremstyle{plain}
\newtheorem{theorem}{Theorem}

\newtheorem{proposition}{Proposition}

\theoremstyle{definition}

\theoremstyle{remark}

\usepackage[font=small,labelfont=bf,labelsep=period]{caption}

\usepackage[numbers]{natbib}

\usepackage{titling}
\begin{document}


\title{\bfseries Personalized Imputation in metric spaces via conformal prediction: Applications in Predicting Diabetes Development with Continuous Glucose Monitoring Information}
\author{Marcos Matabuena\thanks{} \\
\small Health Research Institute of Santiago de Compostela, Santiago de Compostela \\
\small Department of Biostatistics, Harvard University, Boston, MA 02115, USA \\
\small \href{mailto:mmatabuena@hsph.harvard.edu}{mmatabuena@hsph.harvard.edu} \\
\and Carla Díaz-Louzao \\
\small Department of Mathematics, University of A Coruña, Spain \\
\and Rahul Ghosal \\
\small Department of Epidemiology and Biostatistics, University of South Carolina, USA \\
\and Francisco Gude-Sampedro \\
\small ISCIII Support Platforms for Clinical Research,\\ \small Health Research Institute of Santiago de Compostela (IDIS), Santiago de Compostela, Spain\\ \small Concepción Arenal Primary Care Center, Santiago de Compostela, Spain \\ \small University of Santiago de Compostela, Spain}
\date{}

\maketitle

\begin{abstract}
The challenge of handling missing data is widespread in modern data analysis, particularly during the preprocessing phase and in various inferential modeling tasks. Although numerous algorithms exist for imputing missing data, the assessment of imputation quality at the patient level often lacks personalized statistical approaches. Moreover, there is a scarcity of imputation methods for metric space based statistical objects. The aim of this paper is to introduce a novel two-step framework that comprises: (i) a imputation methods for statistical objects taking values in metrics spaces, and (ii) a criterion for personalizing imputation using conformal inference techniques. This work is motivated by the need to impute distributional functional representations of continuous glucose monitoring (CGM) data within the context of a longitudinal study on diabetes, where a significant fraction of patients do not have available CGM profiles. The importance of these methods is illustrated by evaluating the effectiveness of CGM data as new digital biomarkers to predict the time to diabetes onset in healthy populations. To address these scientific challenges, we propose: (i) a new regression algorithm for missing responses; (ii) novel conformal prediction algorithms tailored for metric spaces with a focus on density responses within the 2-Wasserstein geometry; (iii) a broadly applicable personalized imputation method criterion, designed to enhance both of the aforementioned strategies, yet valid across any statistical model and data structure. Our findings reveal that incorporating CGM data into diabetes time-to-event analysis, augmented with a novel personalization phase of imputation, significantly enhances predictive accuracy by over ten percent compared to traditional predictive models for time to diabetes.
\end{abstract}
\section{Introduction}
Recent technological advancements are providing new scientific opportunities in biological measurement systems \cite{hughes2023digital}. Consequently, the emerging novel medical tests enable the monitoring of patient conditions in real-time with high-resolution data \cite{hoeks2011real}. This progress has catalyzed the evolution of clinical systems towards precision and digital health paradigms  \cite{ho2020enabling, tyler2020real}. The next step involves the development of a large number of statistical models to exploit the inherent complexity of these new data structures and support decision-making in the paradigm of personalized medicine \cite{matabuenacontributions}.

One important example of recent technological advancements is seen with continuous glucose monitoring (CGM) devices \cite{freckmann2020basics}. Nowadays, these devices are designed to be minimally invasive and enable the detailed tracking of glucose levels at regular intervals over extended periods, including weeks and months \cite{matabuena2023reproducibility}. This provides a comprehensive view of the temporal dynamics of an individual's glucose metabolism \cite{garg2023past}. Originally developed to significantly improve the management of potentially dangerous situations, such as low blood sugar episodes in people with type 1 diabetes (hypoglycemia), CGM devices have also proven to be particular useful for managing and monitoring blood sugar levels in individuals with type 2 diabetes. 

With the increased affordability and enhanced accuracy of non-invasive glucose measurement technologies, their adoption is expanding among healthy populations. In the realm of personalized nutrition, CGM devices are pivotal, facilitating the identification of optimal dietary choices through monitoring real-time glucose fluctuations. Moreover, CGM devices have found application in epidemiological research on  medical cohorts composed on health individuals, proving their utility in pinpointing individuals at elevated risk of developing diabetes mellitus \cite{klonoff2023use}. Recent studies have highlighted the advantage of leveraging long-term glucose trends, as reflected by glycosylated hemoglobin levels, over traditional diabetes biomarkers within a general population sample \cite{lau2020hba1c, matabuena2022kernel}. However, the potential of CGM to predict diabetes incidence and the timing of disease onset in non-diabetic populations is yet to be fully explored. This gap underscores the necessity for efficient imputation strategies within two-step study designs, where only a subset undergoes detailed medical assessments, including CGM monitoring, to address this issue.

This paper delves into a pertinent scientific inquiry, aiming to develop a robust clinical score for diabetes prediction that incorporates the distinct glucose profile captured by CGM technology. We utilize data from the Spanish longitudinal diabetes study, AEGIS \cite{matabuena2022kernel, pazos2022aging}, adopting a two-step experimental design with a comprehensive ten-year follow-up \cite{woodward2013epidemiology}. This methodological approach distinguishes our work from other studies that do not incorporate CGM data \cite{edlitz2020prediction,artzi2020prediction}. Historically, the prohibitive cost of CGM devices limited baseline data collection to a secondary subsample of 580 individuals from a larger cohort of 1,516 randomly selected from the general population. The ongoing advancements and cost reductions in CGM technology anticipate its widespread use in healthy demographics, potentially integrating these devices into routine public health diabetes screenings \cite{keshet2023cgmap, 10.1371/journal.pbio.2005143}. This evolution underscores the significance of our novel predictive models based on CGM data.
 \begin{figure}[ht]
\centering
\begin{subfigure} {0.8\linewidth}
\includegraphics[width=\linewidth]{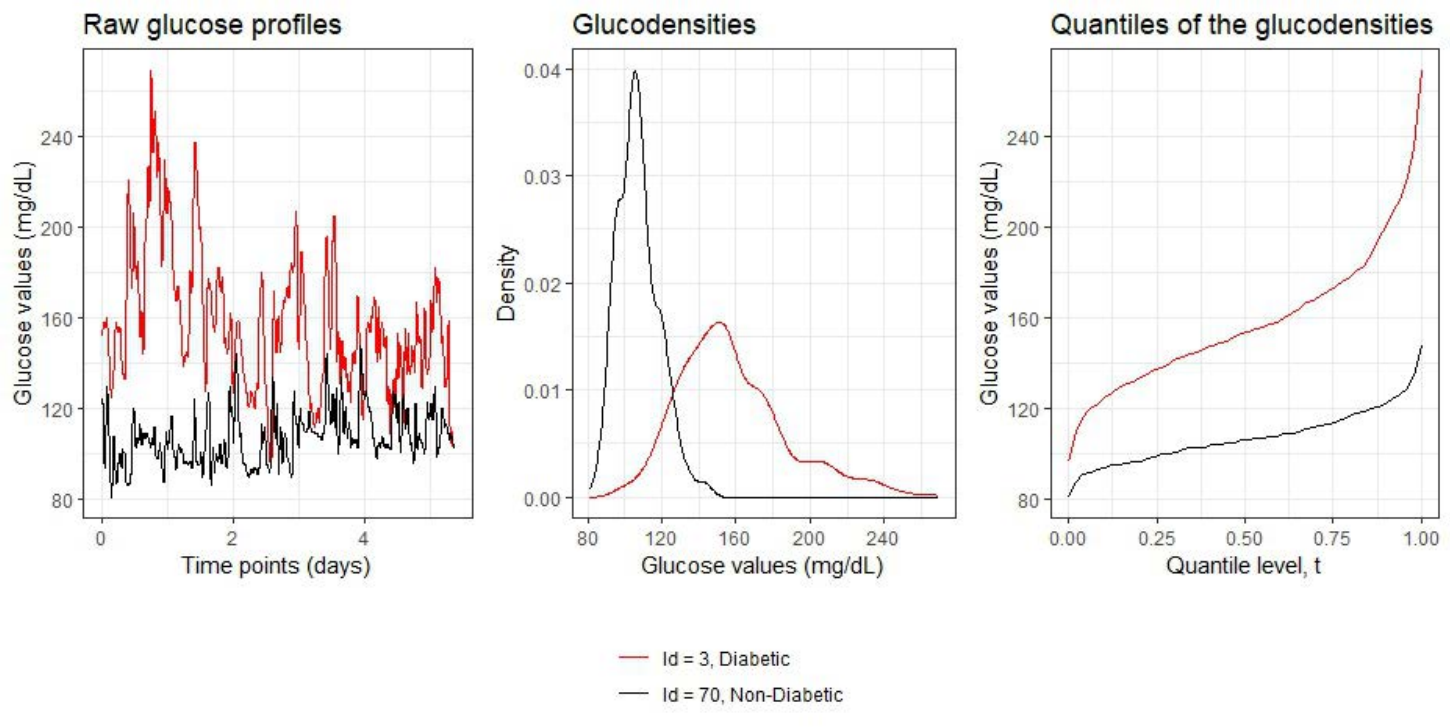}
            \caption{\footnotesize Glucodensity profiles from raw CGM data for a diabetic and non diabetic individual.}
            \label{fig:glucodensities_1}
            \end{subfigure}
           \begin{subfigure} {0.8\linewidth}
\includegraphics[width=0.8\linewidth,height=0.6\linewidth]{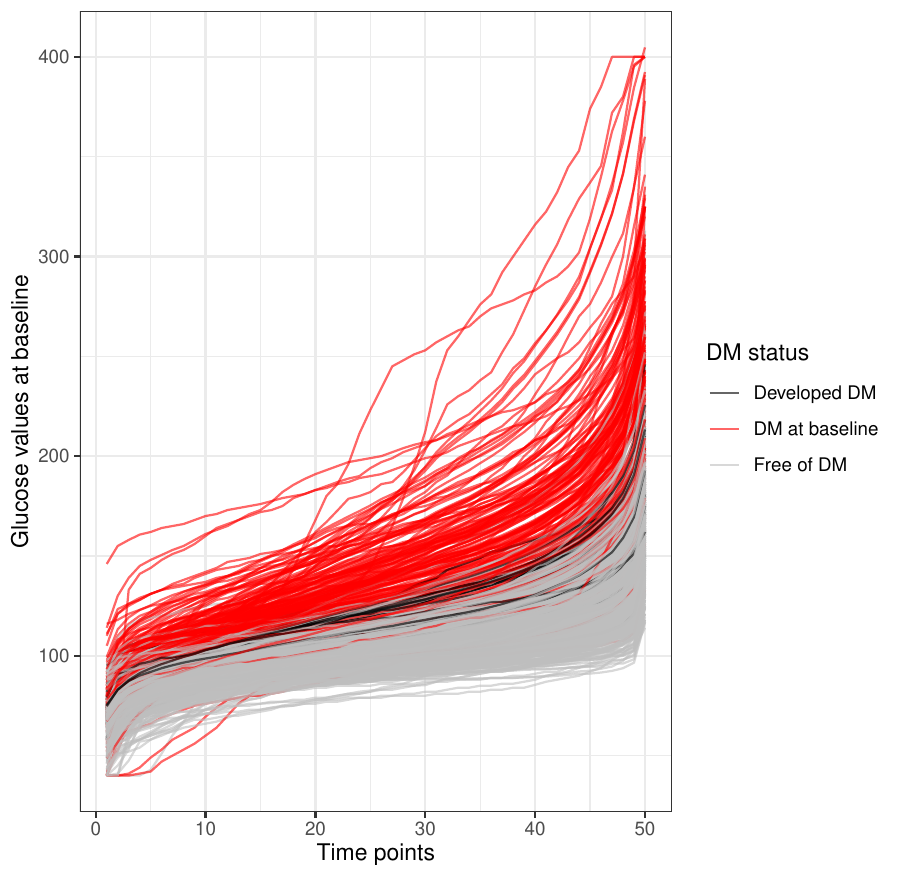}
            \caption{\footnotesize Glucodensities profiles of all subjects with CGM, separated according to the status of  diabetes. Red: individuals with diabetes at baseline. Black: individuals without diabetes at baseline who developed diabetes throughout the study. Grey: individuals free of diabetes at the end of the study.}
            \label{fig:glucodensities_basal}
        \end{subfigure}
\end{figure}

The concept of glucodensity provides more information than traditional CGM summaries \cite{matabuena2021glucodensities}. Figure \ref{fig:glucodensities_1} (top panel) shows the glucodensity profiles of a randomly chosen diabetic and non-diabetic subject, while the bottom panel displays these profiles across all subjects, categorized into three groups: those with pre-existing diabetes, those who developed diabetes during the study, and those free of diabetes at study's end. We hypothesize that glucodensity offers particular utility in non-diabetic populations, a stark contrast to traditional CGM composicional metrics that are specifically tailored for disease-populations. Unlike these traditional CGM metrics, which are defined by specific glucose thresholds applicable to diabetics, our method is designed to discern subtle differences in glucose homeostasis over the complete range of CGM values recorded by CGM monitor.

We are interested in improving the reliability and power of predictive models for the time to diabetes \cite{Wang2024}. For this purpose, we introduce an imputation step for CGM information that minimizes the impact of using a two-step design \cite{tsiatis2007semiparametric} in terms of statistical efficiency. Generally, in the field of functional data settings, there is a significant gap in the literature concerning the imputation of statistical objects, even for Hilbert space-valued random variables \cite{hsing2015theoretical}. Inadequate imputation can severely affect predictive models  \cite{tsiatis2007semiparametric} due to the large dimensionality of statistical functional objects. To address this gap in the literature, our study introduces a novel metric space imputation framework based on a weighted least squares global Fréchet model \cite{petersen2019frechet}, incorporating a conformal prediction \cite{vovk2005algorithmic} step for robust uncertainty quantification. The incorporation of uncertainty quantification steps provides the opportunity to assess the imputation quality and offer personalized imputation rules in line with precision medicine principles \cite{tsiatis2019dynamic} This method is applicable not just to glucodensity data \cite{matabuena2021glucodensities} but also to other complex statistical objects in separable metric spaces $\Omega$  \cite{petersen2019frechet}.

In diabetes research, risk scores like FINDRISC \cite{makrilakis2011validation} and GDRS \cite{muhlenbruch2018derivation} have been developed using logistic or Cox regression models with scalar lifestyle and demographic variables. However, the integration of CGM data for long-term diabetes onset prediction in healthy populations remains underexplored due to the scarcity of extensive long-term CGM cohorts. Our personalized framework utilizes novel distributional glucodensity representations \cite{matabuena2021glucodensities}, enhancing the prediction of diabetes onset and providing a more comprehensive understanding of glucose dynamics and progression than traditional screening methods for diabetes mellitus disease.



\subsection{Contributions}

We briefly summarize the main methodological contributions of this paper as well as the key findings from the analysis of the AEGIS study for modelling time to diabetes.

\begin{itemize}
   \item To the best of our knowledge, we propose the first global Fréchet regression model for metric spaces with missing responses. Our new estimators, based on inverse-probability weighted estimators, are utilized to impute missing glucodensity. Additionally, we provide an algorithm for estimating the conditional variance of the quantile-based glucodensity representation, assessing the uncertainty resulting from the imputation.
   
   \item We extend conformal inference algorithms to provide prediction regions for distributional representations and define a personalized imputation criterion based on the uncertainty related to the imputation. To the best of our knowledge, this is the first work validating the quality of personalized imputations for functional and distributional responses.
   
   \item Utilizing the personalized imputation tools in the AEGIS study, we provide the following scientific insights:
    \begin{enumerate}
       \item We impute the distributional representations of continuous glucose monitoring (CGM) data using the proposed global Fréchet regression model.
       
       \item With our conformal inference algorithm, we identify patients in whom glucodensity can be imputed for follow-up and time-to-diabetes analysis in this longitudinal study.
       
       \item In the  time-to-event analysis of diabetes, we demonstrate that our glucodensity approach, aided by personalized imputation, outperforms traditional CGM metrics. The proposed models show superior prediction accuracy compared to existing literature, highlighting the effective incorporation of CGM data as a reliable source of information for the progression of diabetes mellitus.
    \end{enumerate}
    \item The codes, along with the methods proposed and utilized in this study, are available for reproducibility purposes. They are publicly accessible at \url{https://github.com/CarlaDiaz/Conformal_Imputation}.

\end{itemize}





\section{Background and related work}

\subsection{Statistical Models in Metric Spaces}

One of the most prominent applications of statistical modeling in metric spaces is in biomedical problems, particularly in personalized and digital medicine \cite{rodriguez2022contributions}. These applications often involve complex statistical objects, such as curves and graphs, to record physiological functions and measure brain connectivity patterns at a high resolution. A notable example is the concept of "glucodensity" \cite{matabuena2021glucodensities}, a distributional representation of glucose profiles. This concept has significantly advanced diabetes research methodologies \cite{matabuena2022kernel} and has proved useful in analyzing accelerometer data \cite{matabuena2021distributional,ghosal2021scalar,ghosal2023multivariate}. Methodologically, statistical regression analysis in metric spaces is an emerging field \cite{fan2021conditional,chen2021wasserstein,petersen2021wasserstein,tucker2022modeling,ghosal2021fr,petersen2016functional,zhou2021dynamic,dubey2022modeling,10.3150/21-BEJ1410,kurisumodel,chen2023sliced}. Recent publications have explored hypothesis testing \cite{10.1214/20-AOP1504,dubey2019frechet,petersen2021wasserstein,fout2023fr}, variable selection \cite{tucker2021variable}, multilevel models \cite{bhattacharjee2023geodesic}, dimension-reduction \cite{zhang2022nonlinear}, semi-parametric \cite{ghosal2021fr,bhattacharjee2021single,ghosal2023predicting}, and non-parametric regression models \cite{schotz2021frechet,hanneke2022universally,bulte2023medoid,bhattacharjee2023nonlinear}.

\subsection{Missing Data Imputation and Statistical Methods for Two-Sample Design Studies}
The treatment of missing data, a longstanding issue in statistics, significantly impacts medical study reliability, as emphasized by leading medical journals \cite{little2012prevention}. However, research addressing missing data in metric space models remains scarce. In spaces embeddable in separable Hilbert spaces (negative-type spaces) \cite{lyons2013distance,lyons2020strong}, we have proposed new statistical hyphotesis for randomized clinical trials and paired design \cite{doi:10.1080/00031305.2023.2200512,matabuena2022kernel}.  For a standard functional data,  current methods primarily utilize functional principal component analysis \citep{yao2005functional,preda2010nipals,chiou2014functional,crambes2019regression} and multiple imputation \citep{he2011functional,rao2021modern}. However, these approaches often inadequately address the uncertainty induced by imputation. For distributional data \citep{ghosal2021distributional,matabuena2021distributional}, it's crucial to account for the constraints of the underlying space of the distributional object.

Recent studies in standard settings have focused on addressing missing data in large cohorts and high-dimensional data, emphasizing the importance of uncertainty quantification \cite{zaffran2023conformal}, dimensionality reduction \cite{zhu2022high}, and imputation step \cite{liu2023kernel}. Machine learning algorithms, such as XGBoost  \cite{doi:10.1080/10618600.2023.2252501}, along with optimal transport-based algorithms \cite{muzellec2020missing}, have shown promise in imputation tasks, proving to be more efficient in certain non-linear settings. In the field of digital medicine, new methods have been developed for wearable data, such as data from accelerometers, using specialized models that focus on aggregate summaries like physical activity counts \cite{tackney2021framework}. Recent advances in two-sample designs have included proposals for optimal subsampling methods and efficient influence function-based estimation techniques \cite{lumley2022choosing, https://doi.org/10.1002/sim.9300}. To date, specific studies on imputing functional data, especially density functions in non-vectorial spaces, remain unexplored. The development of specific methods for functional data, such as medical images and distributional representations for wearable information, is increasingly relevant in precision medicine for the proliferation of summarizing the medical conditions of patients with complex statistical objects \cite{matabuenacontributions}.
\section{AEGIS Study Overview and CGM Glucodensity Approach}
\subsection{Study Background}
The A Estrada Glycation and Inflammation Study (AEGIS) \cite{gude2017glycemic}, spanning over a decade, longitudinally tracks 1516 subjects to explore health dynamics, focusing specifically on diabetes. A distinctive aspect of AEGIS is the adoption of continuous glucose monitoring (CGM) technology, offering in-depth glucose profiles for a significant subset of participants of health individuals  at two  time points (years 0 and 5), which is a notable deviation from many clinical studies with shorter durations and fewer participants.

\subsection{Study Goals}
AEGIS aims to: 
a) Identify biomarkers within CGM data for stratifying diabetes risk and complications.
b) Develop dynamic patient phenotypes based on glucose evolution.
c) Characterize metabolic changes to enhance personalized clinical interventions in diabetes research.

\subsection{Data Collection and Participants}
CGM data were recorded every 5 minutes, encompassing about one-third of the study's participants at various time points. Baseline data include dietary habits, laboratory values, and questionnaires assessing metabolic capacity and lifestyle factors. Out of the 1516 participants, 622 were selected for CGM procedures, with 580 successfully completing the protocol and providing analyzable data.

\subsection{Data Analysis Goals}
The focus of this paper is to establish a novel diabetes risk model using high-resolution CGM data from a 9-year longitudinal follow-up. This model, formulated using baseline data, aims to highlight the superiority of CGM in comparison to traditional diabetes biomarkers such as A1C and FPG.

\subsection{CGM Data Collection Protocol}
Participants were equipped with an \textit{Enlite\texttrademark} sensor and \textit{iPro\texttrademark} CGM device (Medtronic Inc., Northridge, CA, USA). Glucose concentrations were recorded at 5-minute intervals for 7 days. The analysis omits the first day and any day with over 2 hours of data-acquisition failure.

\subsection{Ethical Considerations}
The study, sanctioned by the Regional Ethics Committee (Comité Ético de Investigación Clínica de Galicia, code: 2012/025), adhered to the Helsinki Declaration guidelines. Informed consent was obtained in writing from all participants.

\subsection{Glucodensity Approach and Distributional Representations}
\label{ss:quantile_representations}

Building upon our prior work \cite{matabuena2021glucodensities}, this paper introduces the "glucodensity approach" to analyze CGM data, a notable advancement beyond conventional time-in-range metrics in diabetes research. These traditional metrics often categorize glucose levels into fixed intervals, which may not capture individual variations, particularly in non-diabetic cohorts.

\subsubsection{Rationale Behind the Glucodensity Approach}
The glucodensity method offers a refined understanding of glucose profiles by considering the entire distribution of glucose values, rather than just the time within specific ranges. This approach is crucial for unraveling the complexities of glucose dynamics, essential for comprehending the progression to diabetes and its complications.

\subsubsection{Modelling Framework}
For each participant in the study, denoted as the $i$-th participant, we analyze their Continuous Glucose Monitoring (CGM) data, represented as $G_{ij}$ where $j = 1, \dots, n_{i}$. These data points are crucial for capturing the distributional characteristics of glucose levels, which are essential for understanding individual metabolic patterns. 

We focus on the empirical quantile function of each participant's glucose measurements. This function is defined for each participant as $Y_i(\rho) = \hat{Q}_i(\rho)$, where $\rho$ ranges over the interval $[0,1]$. Here, $\hat{Q}_i(\rho)$ denotes the generalized inverse of the empirical cumulative distribution function (CDF) associated with the participant's glucose levels. The empirical CDF, $\hat{F}_{i}(a)$, is given by the proportion of glucose measurements that do not exceed a certain level $a$, mathematically expressed as $\hat{F}_{i}(a) = \frac{1}{n_i} \sum_{j=1}^{n_i} \mathbf{1}\{G_{ij} \leq a\}$, where $G_{ij}$ are the glucose values recorded for the $i$-th participant.

\subsubsection{Implications for Diabetes Prediction}
Employing this distributional perspective on glucose monitoring provides a comprehensive view of an individual's glucose regulation and its deviations from normative patterns. This method enables more precise and personalized risk assessments for diabetes, potentially leading to earlier interventions and improved management strategies, contrasting with traditional diabetes risk models that often rely on singular biomarkers like A1C or FPG.

\subsubsection{Advantages Over Traditional Metrics}
The glucodensity approach \cite{matabuena2021glucodensities} offers several advantages over traditional CGM metrics:
i) It captures a comprehensive view of glucose fluctuations over time, including high and low extremes often overlooked in time-in-range analyses.
ii) It facilitates identification of subtle glucose patterns that might signify early metabolic changes leading to diabetes.
iii) It is adaptable to various populations, including non-diabetic ones, thereby enhancing its utility in preventive medicine.

\subsection{Covariates and Analysis}
Our analysis included a subset of 580 participants with complete CGM data. Covariates included demographic characteristics (age, sex, body mass index), laboratory measurements (lipid profile, liver enzymes). Statistical models were adjusted for these covariates to isolate the effect of glucose dynamics on diabetes risk prediction. A complete list of variables used in the creation of predictive score are provided in Table \ref{tab:variables_in_analysis}.


\begin{table}[!h]
\centering
            \caption{\footnotesize Summary of the predictor variables used in the regression analysis, for the whole sample and separated by the Sex and the fact of having or not CGM at baseline. We represent the means and the standard deviations (in brackets) of the continuous variables (Age, Body mass index (BMI), Glycosilated hemoglobin (HbA1c), Fasting plasma glucose (FPG), Albumin, Insulin) for Men and Women with and without CGM at baseline. For the categorical variables Smoking and Diabetes mellitus , we include the absolute frequency and percentage (in brackets). 
            \label{tab:variables_in_analysis}}
            \hspace*{-1 cm}
            \scalebox{0.80}{
                \begin{tabular}{|l|c|c|c|c|c|c|}
                    \hline
                    & \multicolumn{2}{|c|}{Total sample} & \multicolumn{2}{|c|}{CGM} & \multicolumn{2}{|c|}{Not CGM} \\ 
                    \cline{2-7}
                    & Men & Women & Men & Women & Men & Women \\
                    & (n = 678) & (n = 838) & (n = 220) & (n = 360) & (n = 458) & (n = 478) \\
                    \hline
                    Age & 51.97 (17.58) & 53.09 (17.52) & 47.85 (14.79) & 48.21 (14.48) & 53.96 (18.47) & 56.76 (18.69) \\
                    BMI & 28.56 (4.64) & 27.99 (5.40) & 28.92 (4.74) & 27.71 (5.33) & 28.39 (4.59) & 28.19 (5.46) \\ 
                    HbA1c & 5.65 (0.83) & 5.57 (0.65) & 5.64 (0.88) & 5.52 (0.69) & 5.66 (0.81) & 5.62 (0.62) \\ 
                    FPG & 97.54 (24.50) & 92.05 (21.03) & 97.06 (23.37) & 90.96 (20.85) & 97.79 (25.04) & 92.87 (21.15) \\ 
                    Albumin & 4.48 (0.25) & 4.35 (0.22) & 4.51 (0.23) & 4.36 (0.21) & 4.46 (0.25) & 4.35 (0.22) \\
                    Insulin & 14.02 (13.50) &  11.69 (7.31) & 15.54 (18.41) & 11.69 (7.41) & 13.29 (10.29) & 11.69 (7.24) \\ 
                    \hline
                    Smoking & & & & & & \\
                    \hspace{0.25cm} Ex-smoker & 267 (39.38\%) & 128 (15.27\%) & 77 (35.00\%) & 77 (21.39\%) & 190 (41.48\%) & 51 (10.67\%) \\
                    \hspace{0.25cm} Smoker & 172 (25.37\%) & 124 (14.80\%) & 56 (25.45\%) & 62 (17.22\%) & 116 (25.33\%) & 62 (12.97\%) \\
                    Diabetes mellitus  & 101 (14.90\%) & 82 (9.79\%) & 33 (15.00\%) & 31 (8.61\%) & 68 (14.85\%) & 51 (10.67\%) \\ 
                    \hline
                \end{tabular}
            }
        \end{table}

\section{Mathematical Models}
\label{sec:algorithm}

\subsection{Imputation Step from Distributional Representations}
\label{sec:imputation}
\subsubsection{Linear Regression Model for Metric Space Responses: Global Fréchet Model}

We consider the multivariate random variable $(X, Y)\in \mathcal{X}\times \mathcal{Y}$, where $X \in \mathcal{X}= \mathbb{R}^{p}$ and $Y\in \mathcal{Y}$, with $\mathcal{Y}= (\Omega, d)$ being a separable metric space that adheres to the regularity conditions introduced in \cite{petersen2019frechet}. These conditions ensure the existence and uniqueness of the Fréchet conditional mean denoted as $m(x) = \mathbb{E}(Y|X=x)$ for all $x\in \mathcal{X}$.

In this article, our focus lies on scenarios where the conditional Fréchet mean is expressed through a linear regression model between the predictor and response variables. Such a regression model is known as a global Fréchet regression and is defined as:

\begin{equation}
    m(x) = \arg \min_{y\in \mathcal{Y}} \mathbb{E}\left( \left[1 + (x - \mu)^{\intercal}\Sigma^{-1}(X-\mu)\,\right] d^{2}(Y, y)\right),
\end{equation}

where $\Sigma = \text{Cov}(X,X)$, and $\mu = \mathbb{E}(X)$. Given an i.i.d. (independent, identical, and distributed) random sample $\mathcal{D}_{n}=\{(X_i,Y_i)\}_{i=1}^{n}$, we can construct an estimator $\hat{m}(\cdot)$ from the Global Fréchet model as follows:

\begin{equation}
    \hat m(x)=  \arg  \min_{y\in \mathcal{Y}} \sum_{i=1}^{n} \omega_{in}(x) d^{2}(y,Y_i),
\end{equation}

where $\omega_{in}(x)= \frac{1}{n} \left[1 + (x-\overline{X})^{\intercal}\hat{\Sigma}^{-1}(X_i-\overline{X})\right]$, with $\overline{X}= \frac{1}{n} \sum_{i=1}^{n} X_i$, and $\hat{\Sigma}= \frac{1}{n} \sum_{i=1}^{n} (X_i-\overline{X})(X_i-\overline{X})^{\intercal}$.

\subsection*{Global Fréchet Regression from Weighted Least Squares}

We extend the concept of global Fréchet regression, to incorporate sampling mechanisms induced by missing data patterns, by following scalar response regression models with missing responses. In particular, consider the weighted least squares (WLS) linear regression for scalar response $Y_i \in \mathbb{R}$ . Let $\mathcal{D}_{n}= \{(X_i,Y_i, w_i)\}_{i=1}^{n}$ be the observed random sample, where $w_i$ denotes the weight of participant $i$. For $Y_i \in \mathbb{R}$, under a WLS model, the objective function is given by:

\begin{align*}
    \hat \beta &= \arg  \min_{\beta} \sum_{i=1}^n w_i (Y_i - \langle X_i,\beta \rangle)^2\\
    &= \arg  \min_{\beta} \|\sqrt{W}(Y - X\beta)\|^2,
\end{align*}

where $\beta = (\beta_1, \dots, \beta_p)^{\intercal}$ is a vector of model parameters, $W = \text{diag}(w_1, \dots, w_n)$ is a weight matrix, $Y = (Y_1, \dots, Y_n)$, $X = (X_1, \cdots, X_n)$, and $\|\cdot\|$ is the Euclidean norm. The solution is $\hat \beta = (X^{\intercal} W X)^{-1}X^{\intercal}WY$.

\noindent A future prediction at any $x \in \mathbb{R}^{p}$ is given by:
\begin{align*}
    x\hat \beta &= x(X^{\intercal}W X)^{-1}X^{\intercal}WY\\
    &= k(x)Y\\
    &= \sum_{i=1}^n k_{i,w}(x)Y_i,
\end{align*}

\noindent where $k_{w}(x) = (X^{\intercal} W X)^{-1}X^{\intercal}W$ with $k_{w}(x) = (k_{1,w}(x), \ldots , k_{n,w}(x))$ and $k_{i,w}(x) = \frac{s_{i,w}(x)}{\sum_{i=1}^n s_{i,w}(x)}$. The estimator for conditional mean can be reformulated as:

\begin{equation}
  \hat{m}(x)= \arg  \min_{y\in \mathbb{R}}\sum_{i=1}^n k_{i,w}(x)(Y_i - y)^2.  
\end{equation}

\begin{proposition}
The WLS estimator from the global Fréchet model is:

\begin{equation}
    \hat m(x)=  \arg  \min_{y\in \mathcal{Y}} \frac{1}{n} \sum_{i=1}^{n} \omega_{in}(x)  d^{2}(y,Y_i),
\end{equation}

where

\[
\omega_{in}(x) = \frac{w_{i}\left[1 + (x-\overline{X})^{\intercal}\hat{\Sigma}^{-1}(X_i-\overline{X})\right]}{\sum_{j=1}^{n} w_{j}\left[1 + (x-\overline{X})^{\intercal}\hat{\Sigma}^{-1}(X_j-\overline{X})\right]}
\]

and $\overline{X}= \frac{1}{n} \sum_{i=1}^{n} X_i$, and $\hat{\Sigma}= \frac{1}{n} \sum_{i=1}^{n} (X_i-\overline{X})(X_i-\overline{X})^{\intercal}$.
\end{proposition}
This is directly obtained by extending the  WLS estimation criterion (3) for global Fréchet model.

\subsection*{Linear Regression Model for Missing Metric Space Responses}

\subsubsection*{Estimation with Missing Observations}

In cases where some distributional observations $Y_i$ are missing ($\delta_{i}=0$ for some $i \in \{1,\dots, n\}$), we introduce a special weighting estimator of the form:

\begin{align}
    \hat{m}(x)&= \arg \min_{y\in \mathcal{Y}}  \sum_{i=1}^{n} \omega_{in}(x) d^{2}(y,Y_i),\\
  \omega_{in}(x) &= \frac{ w_{i} \left[1 + (x-\overline{X})^{\intercal}\hat{\Sigma}^{-1}(X_i-\overline{X})\right]}{\sum_{j=1}^{n} w_j  \left[1 + (x-\overline{X})^{\intercal}\hat{\Sigma}^{-1}(X_j-\overline{X})\right]}, w_{i}= \frac{ \frac{\delta_{i}}{\hat{P}(\delta=1|X=X_{i})}}{\sum_{i=1} \frac{\delta_{i}}{\hat{P}(\delta=1|X=X_{i})}}.
\end{align}

In essence, here we perform inverse-probability weighting (IPW) and consider only the non-missing observations ($\delta_{i}=1$) in the construction of regression model. 

\subsection{Closed-Projection Algorithm for 2-Wasserstein Metric}

In the context of our research, each observation indexed by $i=1,2,\dots,n$ represents a patient under study, with $Y_i$ denoting the distribution or the functional outcome corresponding to the $i$-th participant. We proceed by constructing the regression model, directly focusing on modeling the point-wise mean of the quantile function $Y_i(t)$, $t \in [0,1]$ as a function of the covariates. This choice is motivated by the connection of the quantile function to the $2-$Wasserstein distance, and essentially models the Wasserstein barrycenter of the distributional outcome \citep{ghosal2023distributional} based on the covariates.

The $2$-Wasserstein distance, denoted as $d_{W_2}(\mu,\nu)$, serves as a powerful tool for measuring the dissimilarity between probability measures, making it an ideal choice for our analysis. When considering $\mu$ and $\nu$ as two suitable measures on $\mathcal{R}$ with finite second moments, and $Q_\mu$ and $Q_\nu$ as their respective quantile functions, the Wasserstein distance $d_{W_2}(\mu,\nu)$ between $\mu$ and $\nu$ is known to be equivalent to the $L^2$ distance between $Q_\mu$ and $Q_\nu$, as expressed in Equation \ref{eq:wassL2equiv}:

\begin{equation}
    \label{eq:wassL2equiv}
    d_{W_2}(\mu,\nu) = \left[\int_0^1 (Q_\mu(t) - Q_\nu(t))^2\mathrm{d}t\right]^{1/2}.
\end{equation}
This elegant equivalence allows us to bridge the gap between probability measures and quantile functions, offering profound insights into the Wasserstein distance's significance. See \cite{ghosal2021distributional} for various advantages offered by the quantile function based distributional representation as opposed to histogram or densities.

As a consequence, the Fréchet mean \cite{petersen2019} of a random measure can be characterized by the point-wise mean of the corresponding random quantile process. Therefore, by introducing a regression model for the random quantile function $Y_i$, we implicitly construct a model for the conditional Fréchet mean of the underlying glucose distribution measure \cite{petersen2021wasserstein}.

Let $X_{i}\in \mathcal{R}^p$ represent the $p$-dimensional covariate vector. In this scenario, the global linear Fréchet model takes the form:

\begin{equation}
m(X_{i},t)=E(Y_i(t)|X_i)=\alpha(t) + \beta(t)^{\intercal}X_i,\quad  t\in [0,1].
\label{eq:wass_regression_model}
\end{equation} 
Here, $\alpha(t)$ represents the intercept function, and $\beta(t)$ denotes the coefficient function.

Assuming we have access to a sample $\mathcal{D}_{n}=\{\left(X_i, Y_i,w_i \right)\}^{n}_{i=1}$, where $Y_i$ serves as the response quantile function and ${X}_i\in\mathbb{R}^{p}$, we employ the weighted least squares criterion to estimate the parameters. The procedure can be outlined in two steps. Firstly, for any $t \in [0,1]$, we compute the estimates:

\begin{equation}
\left(\hat{\alpha}(t),\hat{\beta}(t) \right)=\underset{a\in \mathcal{R},b\in \mathcal{R}^p}{\operatorname{argmin}}\, \sum_{i=1}^{n} w_i\left[Y_{i}(t)-a-b^{\intercal} X_{i}\right]^{2}.
\label{eq:reg_param_est}
\end{equation}

These estimates lead to the initial fitted quantile functions:

\begin{equation}
Y_{i}^{*}(t)=\hat{\alpha}(t)+\hat{\beta}^{\intercal}(t) X_{i},\quad t \in [0,1].
\label{eq:Yhat_theta_t}
\end{equation}

However, as a function of $t$, it may occur that $Y_i^*(t)$ is not monotonically increasing. Hence we project this fitted value onto the nearest monotonic function in the $L^2[0,1]$ sense, resulting in valid fitted quantile functions $\hat{Y}_i(t)$ \citep{petersen2019frechet}.
This process yields fitted values $\hat{Y}_i(t)$ for any set of observed covariates $X_i$, thus providing valuable insights into the conditional Fréchet mean (based on $2-$Wasserstein metric) within the context of our research.

\subsection{Conformal inference for distributional representation and missing responses}

Conformal inference, a framework for uncertainty quantification in diverse settings, has emerged as a significant tool in statistics, especially in medical applications. Key advantages of conformal prediction include:
i) Providing model-independent prediction regions, ii) Offering non-asymptotic guarantees under broad exchangeability assumptions, iii) Delivering fundamentally non-parametric predictive regions.

This paper introduces a novel algorithm for conformal inference in distributional regression models, tailored for responses lying within a 2-Wasserstein space. This framework facilitates the definition of point-wise residuals and involves predictors in a separable Hilbert space, denoted as $\mathcal{H}$. By leveraging the supremum norm, we streamline computation of prediction regions for response quantile functions, focusing on conditional scenarios with covariates to establish Type II tolerance regions. Practically, we connect the regression model, symbolized by $m(\cdot,\cdot)$, with the conditional mean estimator ($m(X_i,t)$). Our goal is to construct a prediction region, $\mathcal{C}^{\alpha}(X)$, with a confidence level $\alpha$, ensuring $P(Y \in \mathcal{C}^{\alpha}(X)) = 1 - \alpha$. This region either minimizes volume or conforms to specified geometric constraints. Utilizing conformal inference on a random sample, $\mathcal{D}_{n}$, we assure non-asymptotic guarantees: $P(Y \in \widehat{\mathcal{C}}_{n}^{\alpha}(X)) \geq 1-\alpha$, converging to the oracle prediction region as $n \to \infty$.

In this study, we observe $\mathcal{D}_{n} = \{(X_i, Y_i, \delta_i)\}_{i=1}^{n}$, where $i$ indexes patient data, and $\delta_{i}$ indicates missing data. The missingness, contingent on covariates $Y$ (i.e., $Y \perp \delta | X$), disrupts sample exchangeability. Achieving marginal non-asymptotic guarantees hinges on the true missing data weights $w_{i}$, a rarity in practice. However, with increasing $n$, the marginal coverage is guaranteed, assuming precise estimation of the regression model and missing data mechanism.
The core steps of our proposed conformal prediction algorithm are delineated in Algorithm \ref{alg:metd1}. The calibration sample, $\mathcal{D}_{\text{calibration}}$, is pivotal in conformal inference. It is used to calibrate the algorithm and establish necessary confidence levels or significance thresholds for prediction regions, based on nonconformity scores. This calibration ensures the regions accurately reflect the intended uncertainty level and maintain correct coverage probabilities in the non-asymptotic regime.

\begin{algorithm}
    \caption{Conformal Prediction Algorithm for Distributional Responses and Missing Data}
    \label{alg:met1dmissing}
    \begin{algorithmic}[1]
    \State Partition the sample set $\mathcal{D}_{n}$ into three distinct and independent random samples: $\mathcal{D}_{train1}$, $\mathcal{D}_{train2}$, and $\mathcal{D}_{calibration}$.
        \State Estimate the regression function $m(\cdot,\cdot)$ using the global linear Fréchet model (8) as $\hat{m}(\cdot,\cdot)$ using the random sample $\mathcal{D}_{train1}$.
        \State For each observation $i \in \mathcal{D}_{\text{train2}}$ and time point $t \in [0,1]$, do steps 4 ,5 and 6.
        \State Compute the estimated response $\hat{m}(X_i,t)$.
        \State Calculate the residual ${r}_{i}(t) = |Y_{i}(t) - \hat{m}(X_i,t)|$.
        \State Derive the modulation function $\hat{s}(X_i,t)$ from the sample $\{(X_i,r_{i})\}_{i \in \mathcal{D}_{\text{train2}}}$, where $\hat{s}(X_i,t) = \hat{sd}(X_{i},t)$. 

       \State For each observation $i \in \mathcal{D}_{\text{calibration}}$ perform steps 8 and 9.
            \State Define the nonconformity score $R_{i} = \sup_{t \in [0,1]} \frac{|Y_{i}(t) - \hat{m}(X_i,t)|}{\hat{s}(X_i,t)}$.
        
        \State Estimate the empirical distribution $\widetilde{G}^{*}(t)$ as $\frac{1}{\sum_{i \in \mathcal{D}_{\text{calibration}}, \delta_{i} = 1} w_{i}} \sum_{i \in \mathcal{D}_{\text{calibration}}, \delta_{i} = 1} w_{i} \mathbf{1}\{R_i \leq t\}$.
        \State Compute the empirical quantile $\hat{q}_{1-\alpha}$ at level $1-\alpha$.
        \State Construct the prediction region $\widehat{\mathcal{C}} _{n}^{\alpha}(X,t) = [\hat{m}(X,t) - \hat{q}_{1-\alpha} \hat{s}(X,t), \hat{m}(X,t) + \hat{q}_{1-\alpha} \hat{s}(X,t)]$.
    \end{algorithmic}
    \label{alg:metd1}
\end{algorithm}

\begin{theorem}
    
For any function estimator of the regression function $m(\cdot,\cdot)$, $\hat{m}(\cdot,\cdot)$, invariant to permutations, and a random sample $\mathcal{D}_{n}= \{(X_{i}, Y_{i}, w_{i})\}^{n}_{i=1}$ that is exchangeable (assuming knowledge of $w_i$), the prediction region $\widehat{\mathcal{C}}_{n}^{\alpha}(X)$ for a new observation $X$, defined by Algorithm \ref{alg:met1dmissing}, satisfies:

\[
P(Y\in \widehat{\mathcal{C}}_{n}^{\alpha}(X))\geq 1-\alpha
\]

\end{theorem}    

\begin{proof}    
Available in the Supplementary Material.
\end{proof}

\subsection{Personalized Imputation with Conformal Inference}
\label{sec:imput}

This paper's primary objective is to develop a rigorous mathematical framework for assessing the validity of imputed response values in patient data. For each patient $i$, we consider the scenario where the random response $Y_{i}(t)=\hat{Y}_{i}(t)$ for $t\in [0,1]$ is imputed, signified by $\delta_{i}=0$. The appropriateness of the imputation for each patient is evaluated in light of the associated uncertainty, encapsulated by the parameter $\hat{r}_{i}$. Specifically, for each patient $i$ and a given confidence level $\alpha\in (0,1)$, the uncertainty radius $\hat{r}_{i}$ is calculated as the maximum deviation across the interval $[0,1]$, defined by $\hat{r}_{i}= \max_{t\in [0,1]} |\hat{q}_{1-\alpha} \hat{s}(X,t)|,$ where $\hat{q}_{1-\alpha}$ represents the quantile associated with the confidence level $\alpha$, and $\hat{s}(X,t)$ denotes the standard deviation of the imputed values at time $t$.

For a given threshold $\gamma>0$, we define the set of imputed observations as follows:
\begin{equation}
    \mathcal{S}_{\delta}= \{i\in [n]; \delta_{i}=0 \text{ and } \hat{r}_{i} \leq \gamma \}.
\end{equation}
Our goal is to ascertain or estimate the optimal threshold parameter $\hat{\gamma}$ that yields high-quality imputations. This determination is based on an interval quality measure of the response $Y$ or a surrogate outcome $Z$, especially in cases involving binary events (like disease occurrence) or time-to-event responses (such as censored responses).

\subsubsection{Practical Implementation and Model Evaluation}

In practice, we evaluate a set of $m$ threshold values for $\gamma$, denoted as $\gamma^{m}= \{\gamma_1<\gamma_{2}<\dots<\gamma_{m} \}$. For each threshold value $\gamma_{s}$, we assess the performance of a statistical model $T$, which is constructed using observations from the set:
\begin{equation}
    \mathcal{B}_{\gamma_{s}}= \{i\in [n]; \delta_{i}=0 \text{ and } i\in \mathcal{S}_{\gamma_{s}} \}.
\end{equation}
The model at threshold $\gamma_{s}$ is then given by:
\begin{equation}
    T_{\gamma_{s}}=T(\{(X_{i}, \widehat{Y}_{i}): i\in \mathcal{B}_{\gamma_{s}}\} \cup \{(X_{i}, Y_{i}): i: \delta_i=1 \}).
\end{equation}

\subsubsection{Model Selection and Contextual Application}

The model $T$ encompasses a variety of statistical methods, including but not limited to regression models, logistic regression, and survival models. The choice of $T$ is contingent upon the nature of the data and the specific research question, aiming to capture the relationship between covariates $X$ and response variable $Y$ or surrogate outcome $Z$. 

In scenarios predicting binary events, $T$ may be a logistic regression model, whereas for time-to-event data, a survival analysis model may be more appropriate. The selection and evaluation of $T$ are thus context-dependent, guided by the specific objectives and characteristics of the study.

\subsection{Asymptotic Theory for Linear Imputation in Metric Spaces}

To clarify and enhance the presentation of the statistical justification for the consistency of mean imputation within a bounded metric space, denoted by $(\Omega, d)$, we revise the description and notation for improved readability. The foundational equations and assumptions are outlined as follows:

We define the functions:
\begin{equation}
M(\gamma, x) = \mathbb{E}\left[\omega(X, x)d^2(Y, \gamma)\right], \quad M_n(\gamma, x) = \frac{1}{n} \sum_{i=1}^n \omega_{in}(x)d^2(Y_i, \gamma).
\end{equation}
These functions are critical for assessing the effectiveness of mean imputation, with $M$ representing the expected metric deviation squared between observed values and an imputation parameter $\gamma$, and $M_n$ denoting its empirical counterpart based on a sample of size $n$.

The following assumptions are necessary for a fixed $x \in \mathbb{R}^p$:
\begin{itemize}
  \item[(P0)] Both the theoretical and empirical minimizers $\mu_p(x)$ and $\hat{\mu}_p(x)$ are confirmed to exist and be unique, with $\hat{\mu}_p(x)$ being almost surely unique. Moreover, for any $\epsilon > 0$, we ensure that $\inf_{d(\gamma, \mu_p(x)) > \epsilon} M(\gamma, x) > M(\mu_p(x), x)$, guaranteeing a unique minimum.
  \item[(P1)] There exists a lower bound $\epsilon > 0$ for the propensity score $\pi(x)$ for any fixed $x \in \mathbb{R}^p$, ensuring the practical applicability of the propensity score.
  \item[(P2)] The difference between the estimated propensity score $\hat{\pi}(x)$ and the true propensity score $\pi(x)$ diminishes at the rate of $\mathcal{O}_p(n^{-1/2})$, confirming the reliability of the propensity score estimation as the sample size grows.
\end{itemize}

Assumption (P0) is crucial for establishing the consistency of the $M$-estimator $\hat{\mu}_p(x)$, implying that the convergence of $M_n$ to $M$ in the empirical process ensures the convergence of their minimizers. The existence of these minimizers is straightforward if $\Omega$ is a compact set. Assumption (P1) introduces a necessary condition related to the propensity score, and (P2) addresses the accuracy of the propensity score estimate.

\begin{theorem}
Under assumptions (P0) to (P2) and with the condition that $\Omega$ is bounded, for any fixed $x \in \mathbb{R}^p$, the following convergence holds:
\begin{equation}
    d(\hat{m}(x), m(x)) = o_p(1),
\end{equation}
where $\hat{m}(x)$ and $m(x)$ represent the imputed and actual mean values, respectively. This equation demonstrates that the imputed means converge in probability to the actual means as the sample size increases, validating the consistency of the mean imputation method.
\end{theorem}


\section{Simulation Study}
In this Section, we investigate the performance of our proposed distributional imputation method via simulations. For this purpose, we consider the following data generating scenarios.

    \subsection{Generation of the simulated data}
    In this Subsection, we describe the data generation mechanisms for the simulation scenarios. We consider three following scenarios for the missing data mechanism. 
\subsubsection*{Missing data mechanism}
    \begin{itemize}
			\item \textbf{Non-dependent}: The probability that the response is missing does not depend on the covariates. In particular, we set $P(\delta_i=0|X_i)=p_i=0.5$.
   
			\item \textbf{Linear}: The probability that the response is observed (not missing) depends on a linear combination of the covariates as $logit(P(\delta_i=1|X_i))= X_i^T\beta$, where $logit(x)=log(\frac{x}{(1-x)})$. Further, we considered scenarios
   with $p_X\in\{1,2,5\}$, denoting the number of scalar predictors. The models for missing responses ($\delta_i=0$ denotes missing) in each case are as follows:
   \begin{itemize}
				\item \textbf{1 covariate}:
					\[logit(P(\delta_i=1|X_i)) = -0.75 + 1.55X_i\]
				\item \textbf{2 covariates}:
					\[logit(P(\delta_i=1|X_i))= -0.75 + 1.89X_{1i} - 0.37X_{2i}\]
				\item \textbf{5 covariates}:
					\[logit(P(\delta_i=1|X_i)) = -0.75 + 0.82X_{1i} - 0.37X_{2i} + 0.09X_{3i} +0.53X_{4i} + 0.75X_{5i}\]
			\end{itemize}
			
\item \textbf{Non-linear}:  The probability that the response is not missing is a non-linear function of the covariates. The models for missing responses ($\delta_i=0$ denotes missing) for varying number of covariates are:
\begin{itemize}
                \item \textbf{1 covariate}:
                    \[logit(P(\delta_i=1|X_i) = -0.75 + 2.55X_i^2\]
                \item \textbf{2 covariates}:
                    \[logit(P(\delta_i=1|X_i) = -0.75 + 1.89X_{1i}^3 - 0.37X_{2i}^2 + 0.75X_{1i} \]
                \item \textbf{5 covariates}:
                    \[logit(P(\delta_i=1|X_i) = -0.75 + 1.12X_{1i}^3 - 0.37X_{2i}^2 + 1.09X_{3i} +0.1\sin(2\pi X_{4i}) + 0.75\cos(2\pi X_{5i})\]
            \end{itemize}

		\end{itemize}
		
		In all three cases above, the mean probability that the response is missing is $0.5$. Also, the main data generation mechanism for the distributional outcome and scalar covariates is the same for all scenarios and described below. 
		
	\subsubsection*{Data generation mechanism}

		The data generation mechanism considered for the distributional outcome $Y_i(t)$ and the scalar predictors $X_i$ is same across all the scenarios. The scalar covariates are independently generated as $X_{ij} \sim \operatorname{U}[0, 1]$, $j \in \{1, \dots, p_X\}$ and $i \in \{1, \dots, n\}$. Let $T=  [0, 1]$ be the quantile grid, which in this case is composed by 50 equidistant points in $[0,1]$. The distributional response $Y_i(t)$ are generated as
\begin{equation}
    Y_{i}(t) =  \sum_{j = 1}^{p_X}X_{ij}\beta_j(t)+ \frac{\sigma_{lp}}{\sqrt{(SNR)}}\varepsilon_{i}.
\end{equation}
The signal-to-noise-ratio ($SNR$) was set to $SNR = 30$.
Here $\sigma_{lp}$ is the empirical standard deviation of the linear predictor $\sum_{j = 1}^{p_X}X_{ij}\beta_j(t)$, and $\varepsilon_{i} \sim \operatorname{N}(0, 1)$ independently. We set $\beta_j(t)=t$ across all the covariates. Four different sample size $n \in \{500, 1000, 2000, 5000\}$ are considered for this simulation study across all possible combination of scenarios (missing data mechanism and $p_X$).

In Figure~\ref{fig:sim} we display the trajectories of the distributional outcomes for one simulated dataset with sample size 500 for 1 (left panel), 2 (middle panel), and 5 covariates (right panel), which can all noticed to be non-decreasing.
\begin{figure}[!ht]
            \centering
                \includegraphics[width = 1\linewidth, height = 0.5\linewidth]{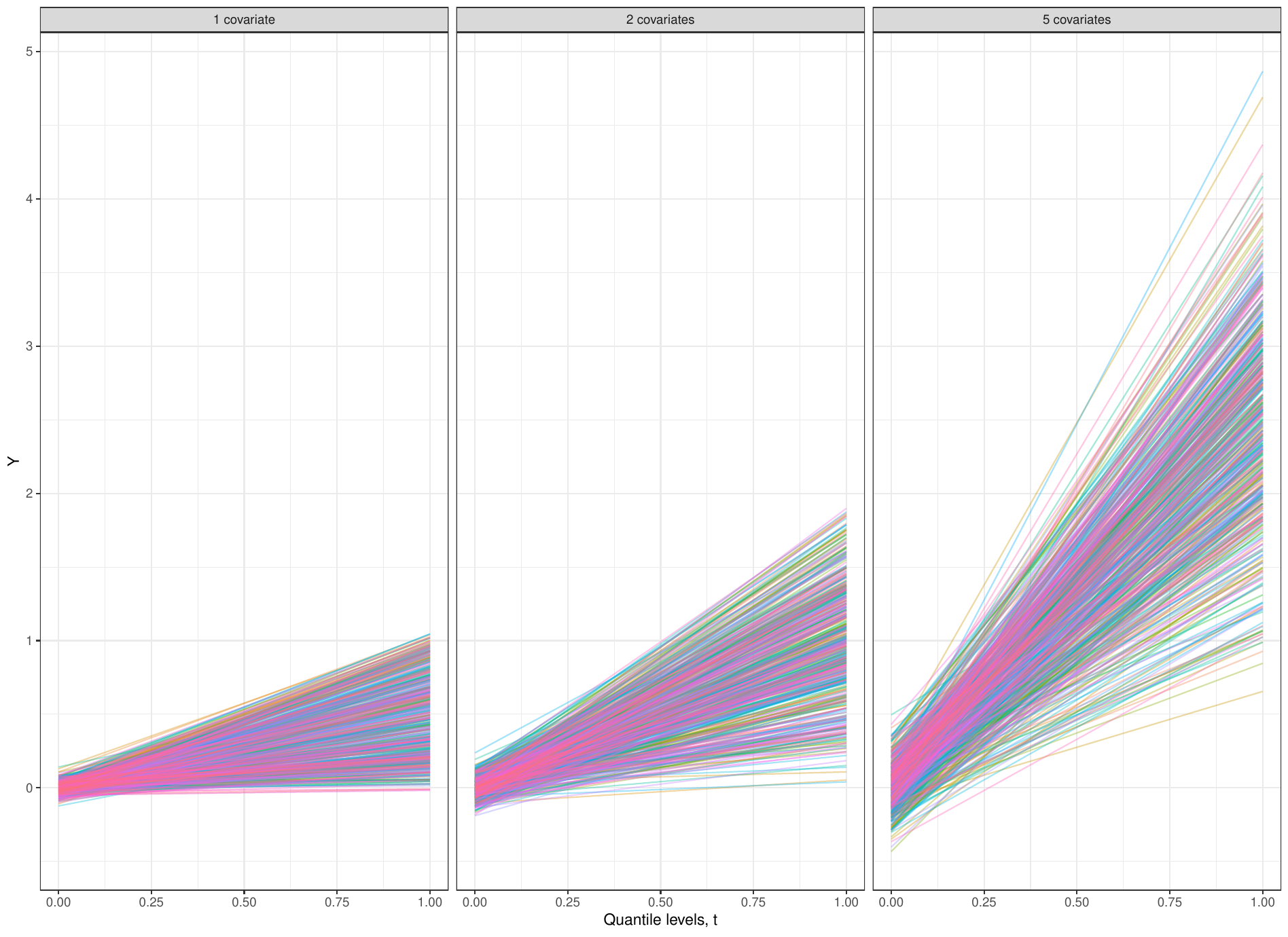}
                \caption{\footnotesize Trajectories of the distributional outcomes $Y_i(t)$ for one simulated dataset with sample size 500 for 1 (left panel), 2 (middle panel), and 5 covariates (right panel).}
                \label{fig:sim}
            \end{figure}

\subsection{Simulation Results}
 As explained in Section 4, the WLS estimator from the the global Fréchet model includes inverse probability weighting in the estimations, these weights being proportional to the probability that the response is not missing. We estimate $\hat{P}(\delta_i=1|X_i)$ using a generalized additive model (GAM) for the binary response $\delta$. The \texttt{mgcv} R package \citep{wood2017generalized} is used for the GAM implementation. Finally, for each subject $i \in \{1, \dots, n\}$, the weights are defined as: 
            \begin{equation}\label{eq:weights_sim}
                w_i = \frac{\mathbf{1}_{\delta_i = 1}}{\hat{P}(\delta_i=1|X_i)} = \begin{cases}
                    0 & \text{if the distributional outcome is missing}, \\
                    \frac{1}{\hat{P}(\delta_i=1|X_i)} & \text{if the distributional outcome is known},
                \end{cases}
            \end{equation}

    Next, we impute the  missing responses using the WLS global linear Fréchet model (8). The final imputed values for the missing responses $\hat{Y}_i(t)$ are obtained using the closed–projection algorithm for $2-$ Wasserstein Metric illustrated in Section 4.2. Finally, we evaluate marginal coverage of the proposed conformal inference algorithm on the dataset with missing responses (where $\delta_i=0$) for a $95\%$ confidence region, i.e., $\alpha=0.05$.

Figure ~\ref{fig:coverage} displays the distribution of estimated coverage across all simulation scenarios. When the mechanism of missing response does not depend on the covariates (Non-dependent scenario, top panel), the median coverage is close to the nominal coverage of $0.95$, regardless of the sample size or the number of covariates. The variability in the estimated coverage, decrease with increasing sample size. In contrast, for scenarios where the missing response mechanism depends on the covariates (linearly or non-linearly), the median coverage is somewhat lower, but always higher than 0.9 and increasing with sample size. The drop in the estimated coverage in these cases is expected, as 
we are using estimated weights instead of actual weights, and it is known that conformal inference is not accurate for cases where the exact weights are unknown and must be estimated. Nonetheless, this reduction is small, particularly for higher sample sizes, and the nominal coverage rate of $95\%$ lies within the two standard error limit of the average estimated coverage. 
            
\begin{figure}[H]
\centering
\includegraphics[width = \linewidth,height = 0.7\linewidth]{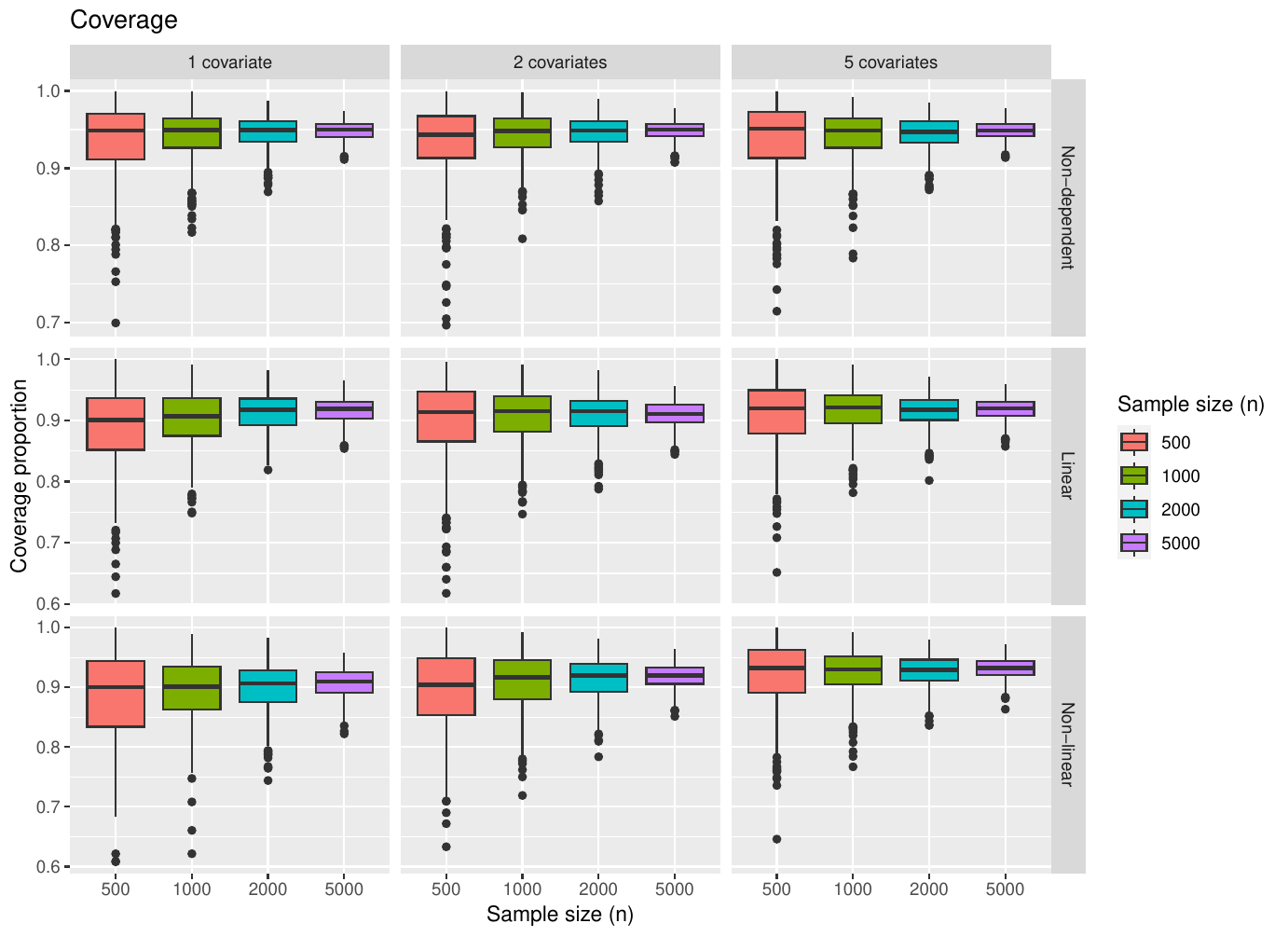}
\caption{\footnotesize Boxplots of the estimated coverage provided by the conformal inference algorithm for 500 M.C replications across every simulation scenario.}
 \label{fig:coverage}
\end{figure}

The prediction performance of the WLS global linear Fréchet model (4) was evaluated by means of in-sample $R^2$ and out of sample Root Mean Squared Error (RMSE). The distribution of $R^2$ across all the simulation scenarios are displayed in Figure ~\ref{fig:r2}. The prediction performance appears to be pretty robust, the median $R^2$ being higher than 0.8, with a low variability with increasing sample size. The more complex the data generating mechanism, the lower the $R^2$, which is expected. 
\begin{figure}[!h]
                \centering
                \includegraphics[width = 0.8\linewidth,height = 0.5\linewidth]{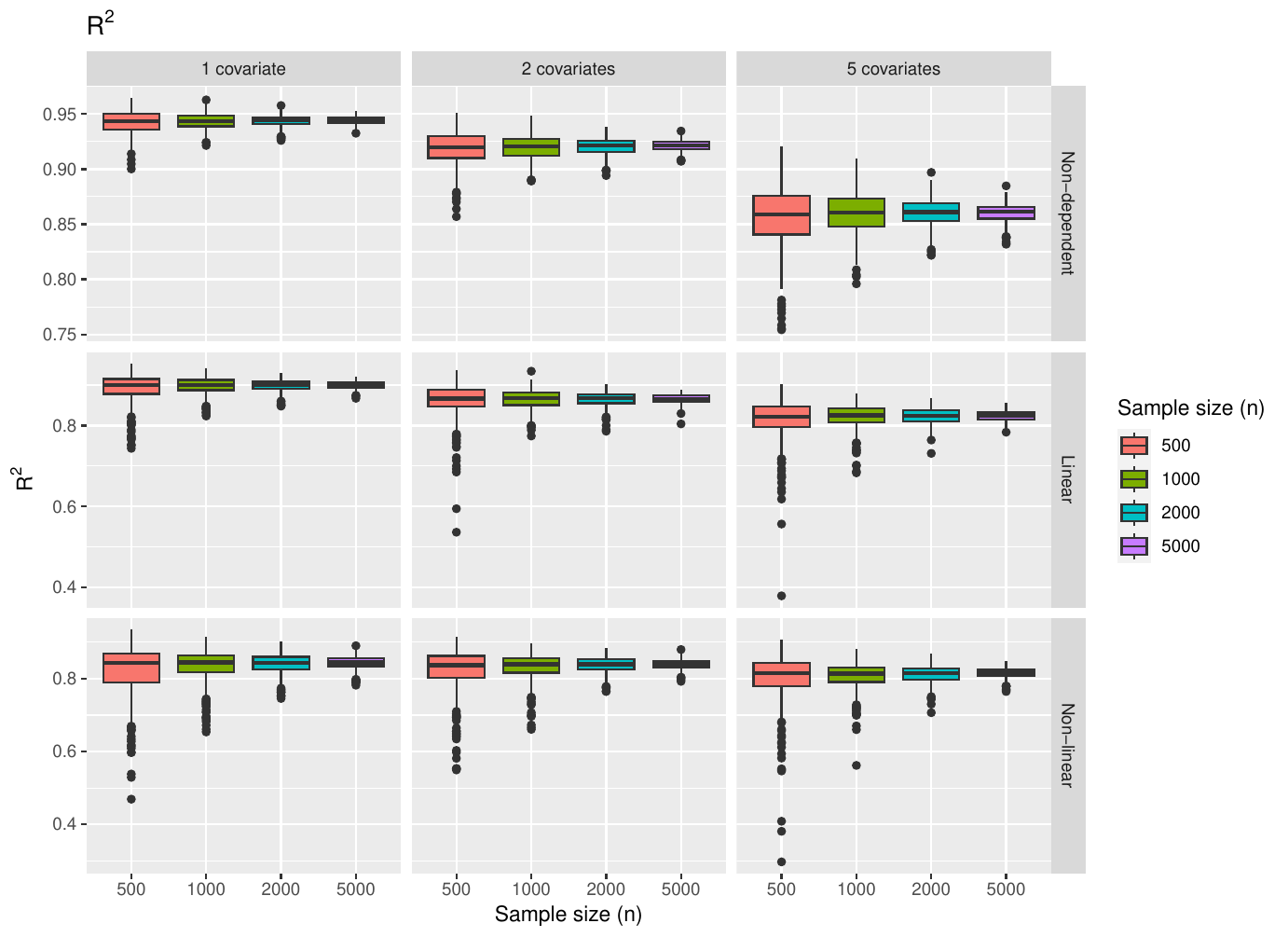}
                \caption{\footnotesize Distribution of the $R^2$ from the global linear Fréchet model across every simulation scenario.}
                \label{fig:r2}
            \end{figure}
The Root Mean Squared Error (RMSE) on the test data ($\delta_i=0$) was computed as, 
            \begin{equation}\label{eq:rmse}
                RMSE = \sqrt{\frac{1}{50|I_{test}|}\sum_{i \in I_{test}}{\sum_{t = 1}^{50}{(Y_{i}(t) - \hat{Y}_{i}(t))^2}}},
            \end{equation}
            \noindent where $|I_{test}|$ is number of subjects in the test set. Distribution of the Root Mean Squared Error (RMSE) are shown in Figure~\ref{fig:rmse}.
            \begin{figure}[!h]
                \centering
                \includegraphics[width = 0.8\linewidth,height = 0.5\linewidth]{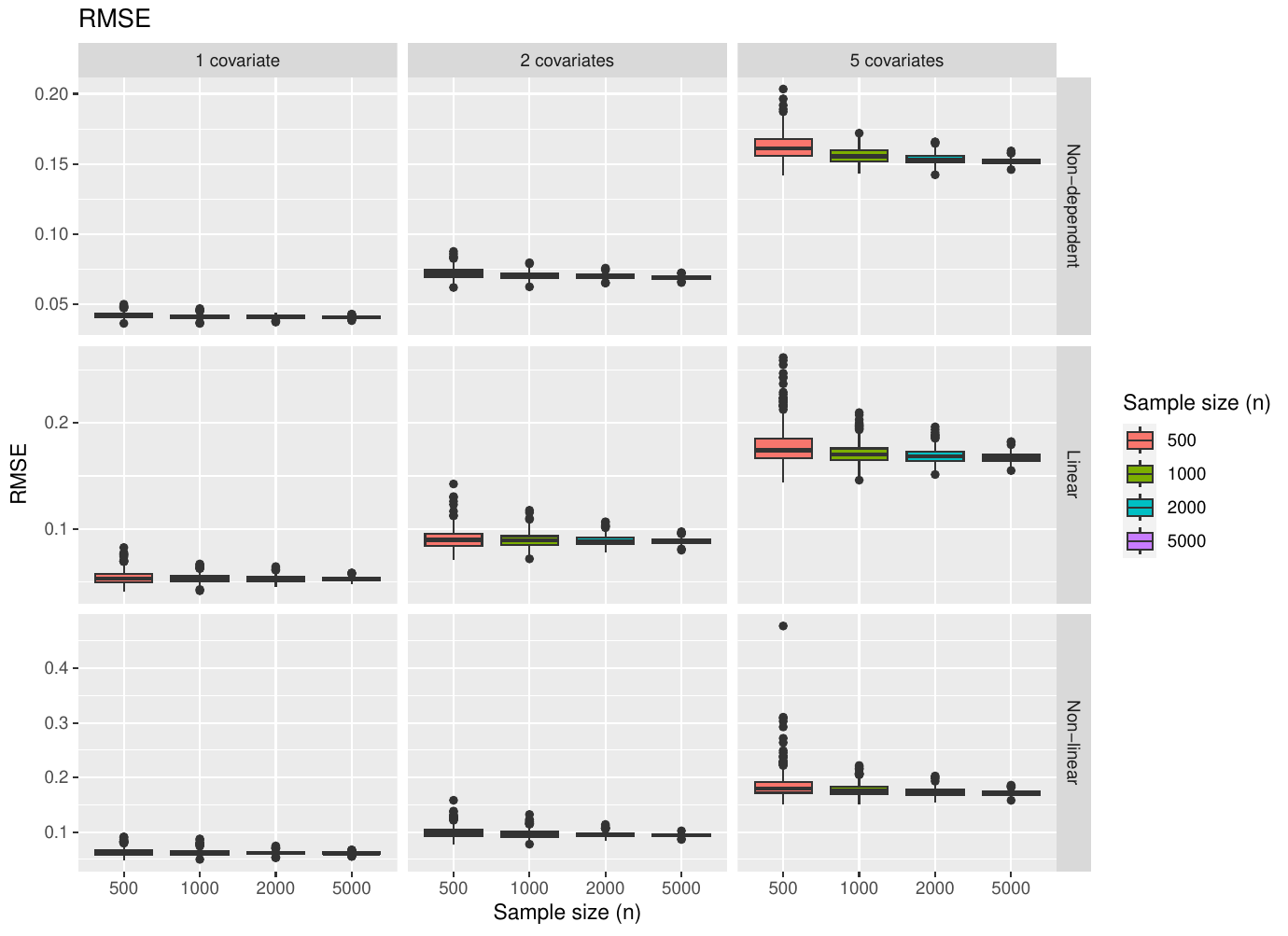}
                \caption{\footnotesize Distribution of the Root Mean Squared Error (RMSE) from the global linear Fréchet model model across every simulation scenario.}
                \label{fig:rmse}
            \end{figure}
Across all cases, the RMSE is small and appears to be decreasing with increasing sample size. Comparing the results across different number covariates, it should be noticed that the variability in the distributional outcome increase with increasing number of covariates as evident from Figure \ref{fig:sim}. Focusing on the sample size $n=500$, for the 1 covariate scenario, the expected RMSE from data generation is approximately 0.041, the estimated median RMSE of the non-dependent scenario is 0.042, while for the linear one it is 0.053 and for the non-linear one (the most complex) it is 0.061. Similarly, for the case of 2 covariates ($n=500$), the expected RMSE was 0.069, the estimated median RMSE across the non-dependent, linear and non-linear scenario were 0.072, 0.090, 0.097 respectively. Finally, for the case of 5 covariates ($n=500$), the expected RMSE was 0.069, the estimated median RMSE across the non-dependent, linear and non-linear scenario were 0.151, 0.161, 0.180 respectively. Overall, we see that the error increase with increase in complexity in the data generation mechanism, but, except for a few outliers, the proposed imputation method provide a robust performance.

\section{Modeling time to diabetes using distributional CGM information}

Diabetes is a complex metabolic disorder where the body struggles to regulate insulin effectively, leading to inconsistent blood sugar levels. It primarily presents in two forms: Type 1 and Type 2. Factors such as genetics, lifestyle, and environmental conditions play crucial roles in the development of diabetes. With an increase in sedentary lifestyles and an aging population, the incidence of diabetes is on the rise. This trend highlights the critical need for innovative public health strategies that focus on early detection and tailored management of the disease.

This section explores the creation of advanced statistical models aimed at predicting diabetes onset by analyzing individual glucose regulation profiles. These profiles are carefully constructed using data from continuous glucose monitoring (CGM) systems over a week. CGM data provides an in-depth analysis of blood sugar variations over time, offering insights beyond conventional health metrics. We utilize a approach called glucodensity, detailed in Section \ref{ss:quantile_representations}, to encapsulate CGM data insights. Nevertheless, CGM data was not available for all participants. To address this, we developed a strategy for imputing missing CGM data, described in Section \ref{sec:imputation}. This method is further refined with personalized adjustments (Section \ref{sec:imput}), improving the accuracy of our prediction model.

Our analytical pipeline has the following steps. Initially, we use a generalized additive model to evaluate the probability of missing data, leading to the calculation of missing data weights (see Section \ref{s:estimation}). Then, we proceed to the imputation phase, incorporating inverse-probability weighting estimators (Section \ref{sec:imputationres}). After imputing data—combining both imputed and original data—we test the model's effectiveness in predicting diabetes onset using Cox models (Section \ref{s:survival}). These models blend CGM and non-CGM data to examine the impact of CGM information on prediction precision through time-dependent ROC curves. By applying personalized imputation criteria (Section \ref{sec:personalized}), we further enhance the model's predictive performance and explore the potential of CGM data in accurately forecasting diabetes onset.

\subsection{Imputing missing Glucodensities using inverse probably weighting estimator}\label{s:estimation}



    \subsubsection{Computing the missing-data weights}

   \begin{figure}[!h]
	   \centering
	   \includegraphics[width = 0.8\linewidth]{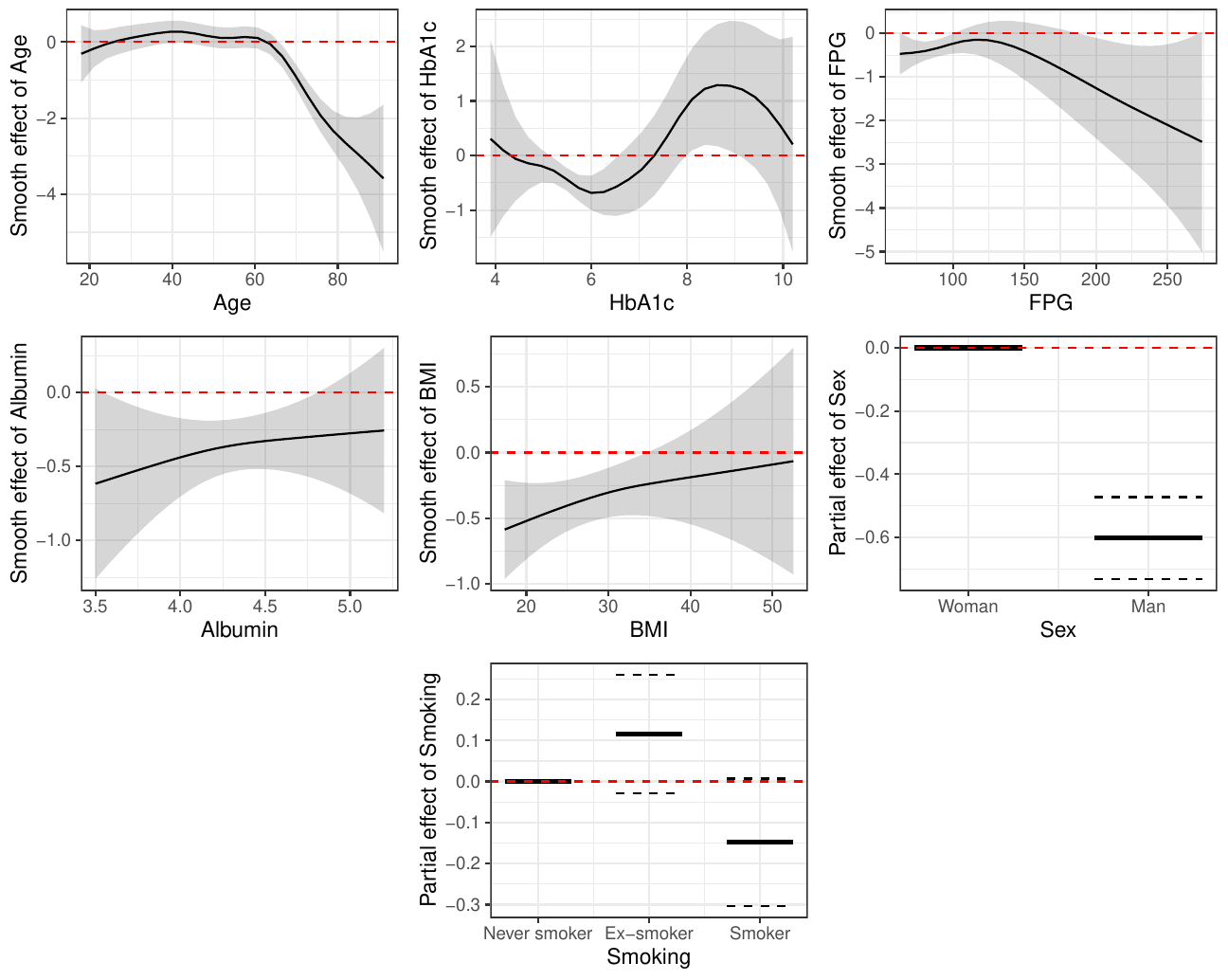}
	   \caption{\footnotesize Smooth effect of ``Age'', ``HbA1c'', ``FPG'', ``BMI'', and ``Albumin'' on the probability of having CGM.}
	   \label{fig:mod_pesos}
        \end{figure}

As a preliminary step in our analysis, we categorize each participant with a binary label: ``0'' indicating the absence of Continuous Glucose Monitoring (CGM) data and ``1'' for those with available CGM data. Interestingly, of the total 1516 subjects in our study, 580 AEGIS-I subjects (\textit{38\%}) were equipped with a CGM device at baseline.

Our next objective is to estimate the probability of the availability of CGM data, considering variables such as Age, Sex, HbA1c, FPG (Fasting Plasma Glucose), Smoking habits, Albumin levels, and BMI (Body Mass Index). To accomplish this, we employ a Generalized Additive Model (GAM) with a logistic link function, as proposed by \cite{wood2017generalized}, which can be mathematically represented as follows:

\begin{equation} \label{eq:mod_pesos}
   \log\left(\frac{\mathbb{P}(CGM=1|X=x)}{1-\mathbb{P}(CGM=1|X=x)}\right) \sim s(\text{Age}) + \text{Sex} + s(\text{HbA1c}) + s(\text{FPG}) + \text{Smoking} + s(\text{BMI}) + s(\text{Albumin}),
\end{equation}

\noindent where $s(\cdot)$ indicates a smooth function for each variable, estimated using thin-plate regression splines, and $CGM \in \{0,1\}$ represents the binary status of CGM usage. This model demonstrate a moderate predictive power, with a deviance explained of $10.14\%$ and an adjusted $R^2$ of 0.11. Figure~\ref{fig:mod_pesos} depicts the influence of these variables, highlighting a notable decrease in CGM participation beyond the age of 65. Conversely, participation likelihood increases for individuals with elevated HbA1c values and is higher among women. These observations suggest that we are dealing with a "Missing Not At Random" (MNAR) data scenario.

Finally, based on the predictions from the above GAM, we compute the weights for the $i$-th patient as follows $
    \omega_{i} = \frac{\delta_{i}}{n \cdot \hat{\pi}(X_{i})}$, where $X_{i}$ denotes the characteristics of the $i$-th patient, and $\hat{\pi}(X_{i})$ represents the estimated conditional probability $\mathbb{P}(CGM=1|X_{i})$.
     
         \subsubsection{Fitting the global Fréchet model from missing responses} \label{sec:imputationres}
 \begin{figure}[ht!]
	   \centering
	   \includegraphics[width = 0.8\linewidth]{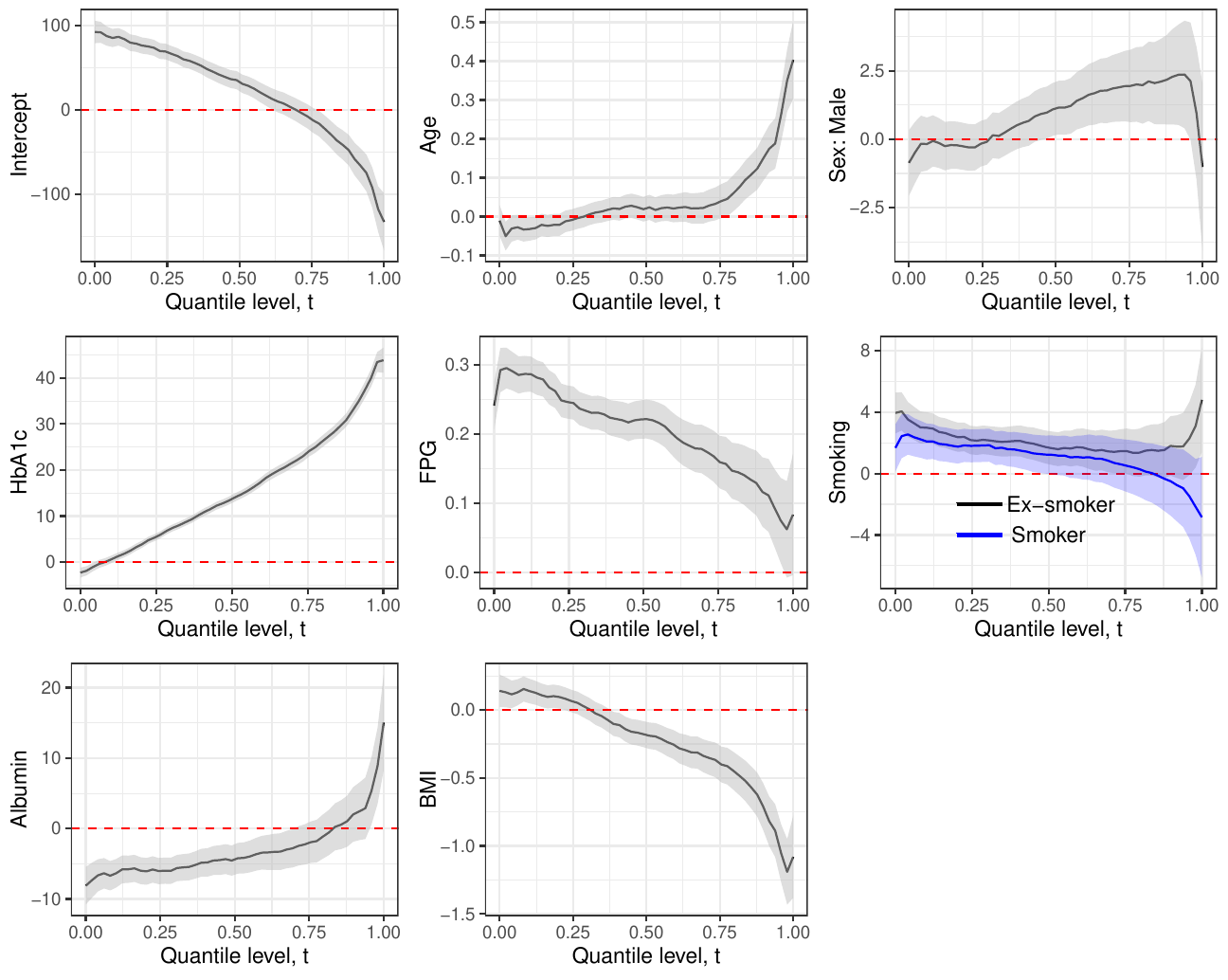}
	   \caption{\footnotesize Estimated functional $\beta-$ coefficients of the Fréchet weighting estimator.}
	   \label{fig:m1}
        \end{figure}

        \begin{figure}[ht!]
	   \centering
	   \includegraphics[width = 0.8\linewidth]{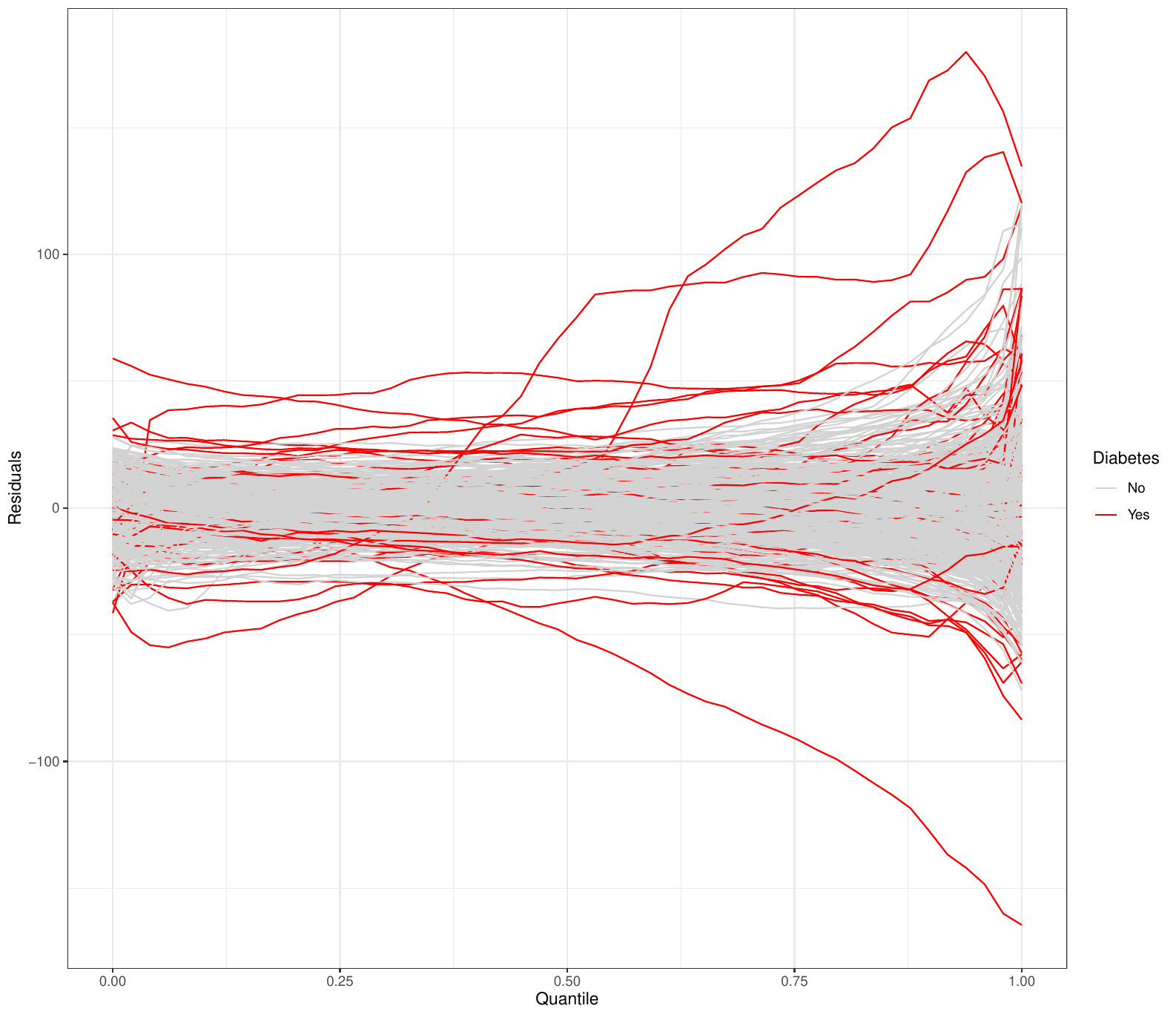}
	   \caption{\footnotesize Residuals for the individuals with CGM measurements. Individuals with diabetes are shown in red.}
	   \label{fig:residuals_diabetes}
        \end{figure}

   For each participant, labeled as the $i$-th individual, we calculate their glucose quantile representation, $Y_{i}(\rho) = \hat{Q}_{i}(\rho)$, over a spectrum of 101 evenly distributed points, $\tau = \left\{ \frac{j}{100} : j = 0, 1, \ldots, 100 \right\}$. This is achieved by leveraging the empirical distribution, $\hat{F}_{i}(t) = \frac{1}{n_{i}} \sum_{j=1}^{n_i} I\{G_{ij} \leq t\}$, which is based on their glucose readings. After establishing the weights and glucodensity quantiles we proceed with the 2-Wasserstein weighted Fréchet regression model based on the covariates: Age, Sex, HbA1c, FPG, Smoking habits, Albumin levels, and BMI. This model demonstrates a coefficient of determination ($R^2$) of 0.63, indicating a significant portion of the variance is explained. The model's coefficients varying over quantile levels are displayed in Figure~\ref{fig:m1}.

Figure~\ref{fig:residuals_diabetes} showcases the residual patterns for each participant with CGM data. Specifically, individuals diagnosed with diabetes at the beginning of the study are marked in red, indicating notably higher residuals for this group. This suggests increased uncertainty in predicting glucodensity quantiles among these participants. Following this, in additional analysis presented in supplementary material, we applied a weighted Fréchet regression model to examine the squared residuals, linking them with the same variables previously considered in the Fréchet regression model. The purpose of this analysis was to probe into the conditional variability, and employ a framework of conformal prediction to enhance our understanding of the underlying patterns.


\subsection{Survival analysis following imputation in AEGIS} \label{s:survival}
In this Section, we go a step further and take the follow-up time into account, to estimate the probability of time-to-diabetes using survival models.  To do so, we are sticking only to the 1293 individuals without diabetes which also have a follow-up time of more than half a year. We used the Fréchet regression model to obtain the CGM-imputed and non-imputed quantile CGM profiles, \(Y_{i}(t)\), for \(t \in [0,1]\) and \(i = 1, \dots, n = 1293\) among which 789 individuals ($61\%$) have imputed CGM. The subjects are followed up approximately over a 10 year period, with median follow-up time being 7.4 years. Among the initial sample 75 individuals developed diabetes ($5.8\%$), by the end of the study. The goal of this subsection is to estimate the risk of diabetes over time as a function of the functional principal components (fPC) of the glucodensities, adjusting or not for other covariates, using an additive Cox regression model. 

        \subsubsection{Survival models with all imputed patients and CGM information}

\begin{figure}[ht!]
	       \centering
	       \includegraphics[width = 0.8\linewidth, height= 0.6\linewidth]{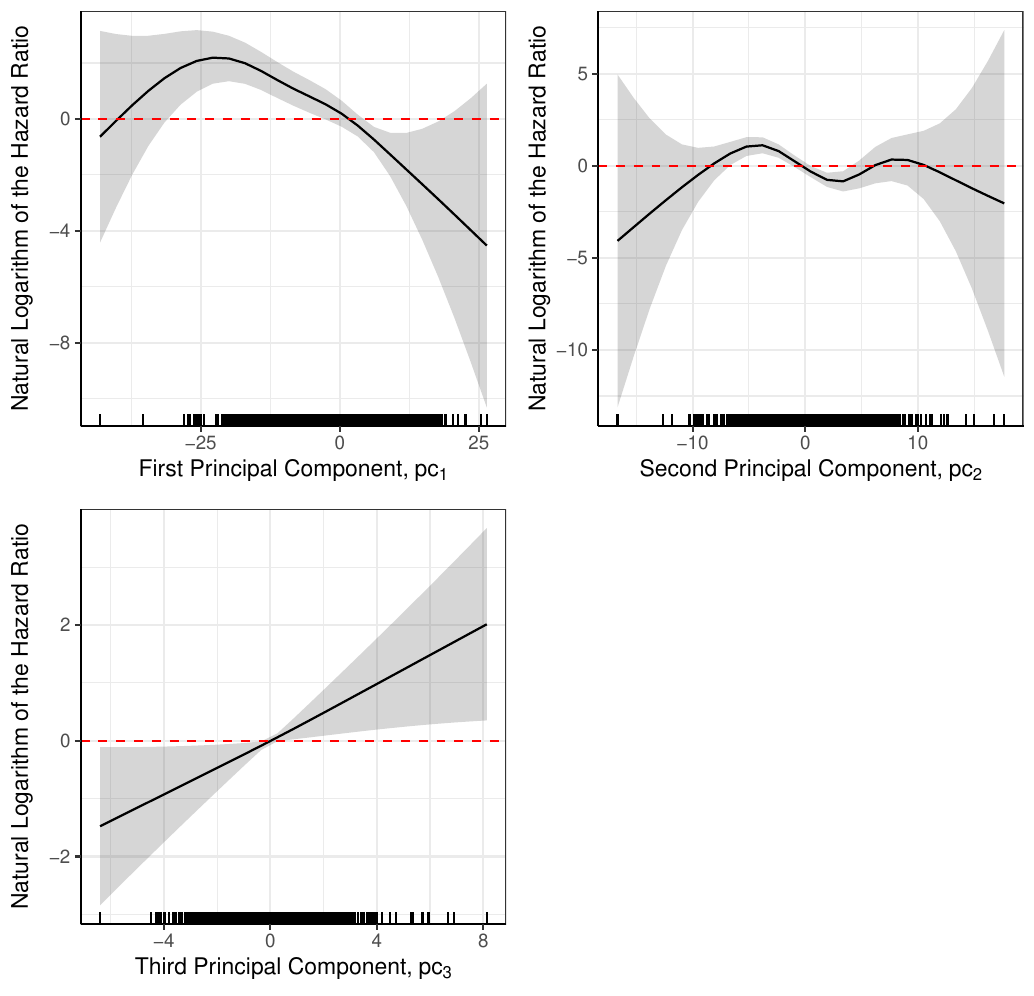}
	       \caption{\footnotesize Smooth effect of the covariates of model that only incorporate CGM  information.}
	       \label{fig:cox_pca_simple}
            \end{figure}

\begin{figure}[ht!]
	       \centering
	       \includegraphics[width = 0.8\linewidth,height= 0.6\linewidth]{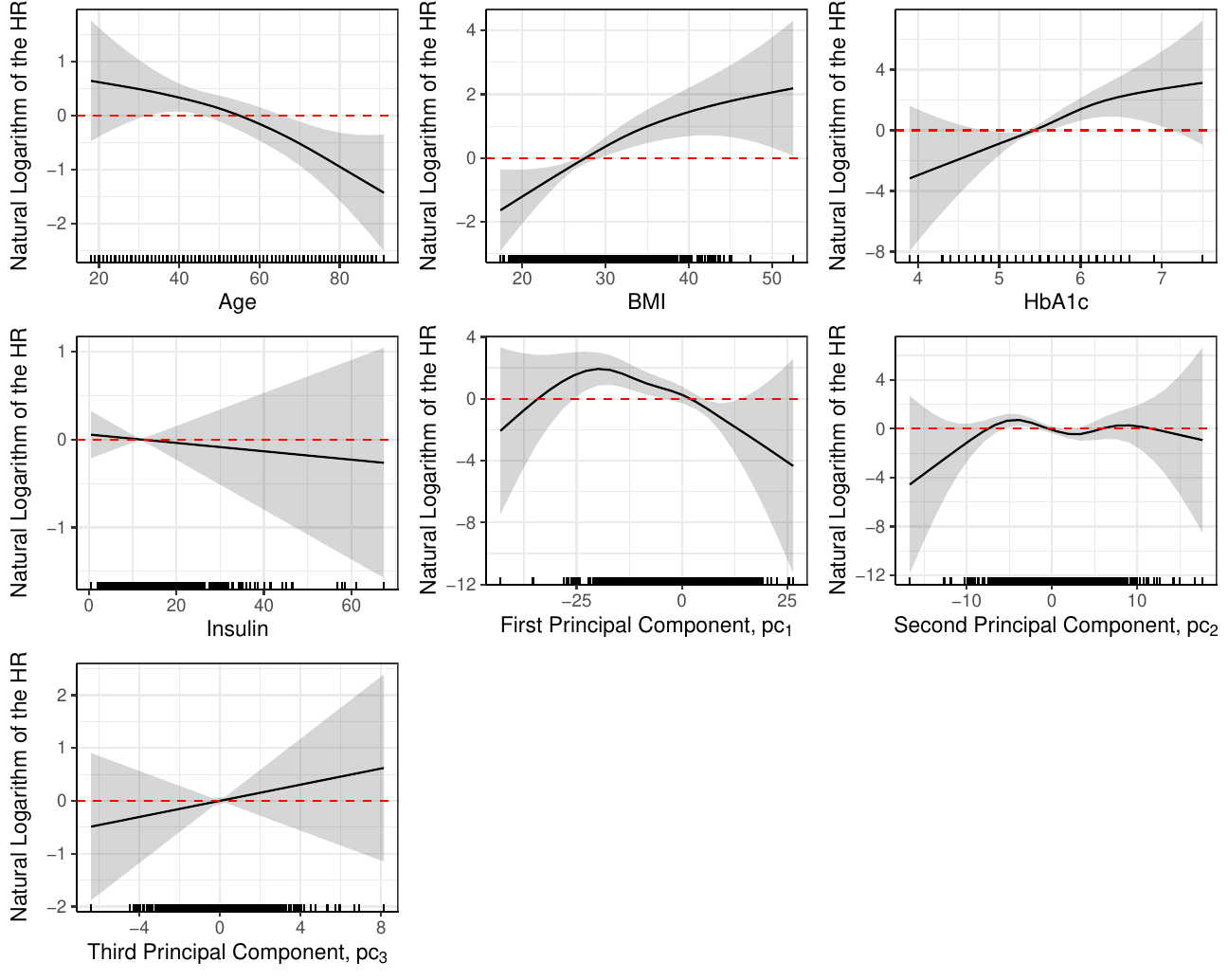}
	       \caption{\footnotesize Smooth effect of the continuous covariates of model with CGM and non CGM information).}
	       \label{fig:cox_pca}
            \end{figure}

In our comparative analysis, we start with a time-to-diabetes model, utilizing Continuous Glucose Monitoring (CGM) data exclusively. Specifically, we use both CGM-imputed and non-imputed quantile CGM profiles, \(Y_{i}(t)\), for \(t \in [0,1]\) and \(i = 1, \dots, n = 1293\), as functional predictors and extract the first three functional principal component (fPC) scores that account for 98\% of the data variability. Subsequently, we adopt an additive Generalized Additive Model (GAM) Cox model as delineated in the `mgcv` package in R, fitting the time to diabetes onset over a 9-year follow-up. We incorporate the three principal components, \(pc_{j}\), \(j = 1,2,3\), in the Cox regression model as follows:
\[
h_{i}(t) = h_0(t) \exp\left\{s(pc_{i1}) + s(pc_{i2}) + s(pc_{i3})\right\},
\]

\noindent where \(h_{i}(t)\) signifies the hazard function over time for the \(i\)-th subject, and \(s(\cdot)\) denotes the smooth function of each covariate. The influence of the three principal component scores on the risk model is displayed in Figure \ref{fig:cox_pca_simple} reveals non-linear dynamics, with the third score, in particular, showing a linear escalation in diabetes risk. An extended model also incorporates the demographic and clinical covariates Sex, Age, BMI, HbA1c, and Insulin—alongside the PCA scores, and is given by
\[
h_{i}(t) = h_0(t) \exp\left\{Sex_i + s(Age_i) + s(BMI_i) + s(HbA1c_i) + s(Insulin_i) + s(pc_{i1}) + s(pc_{i2}) + s(pc_{i3})\right\}.
\]
The estimated effects are displayed in Figure \ref{fig:cox_pca} and underscore the nonlinear effects of PCA scores and reveals an increased diabetes risk associated with higher BMI and HbA1c levels. Intriguingly, the risk diminishes with age, suggesting a protective effect against diabetes in individuals maintaining a non-diabetic state and favorable metabolic health over time.


 \subsubsection{Comparing CGM and non-CGM models in terms of AUC over time}

        \begin{figure}[!h]
	    \centering
	    \includegraphics[width = 0.6\linewidth,height= 0.6\linewidth]{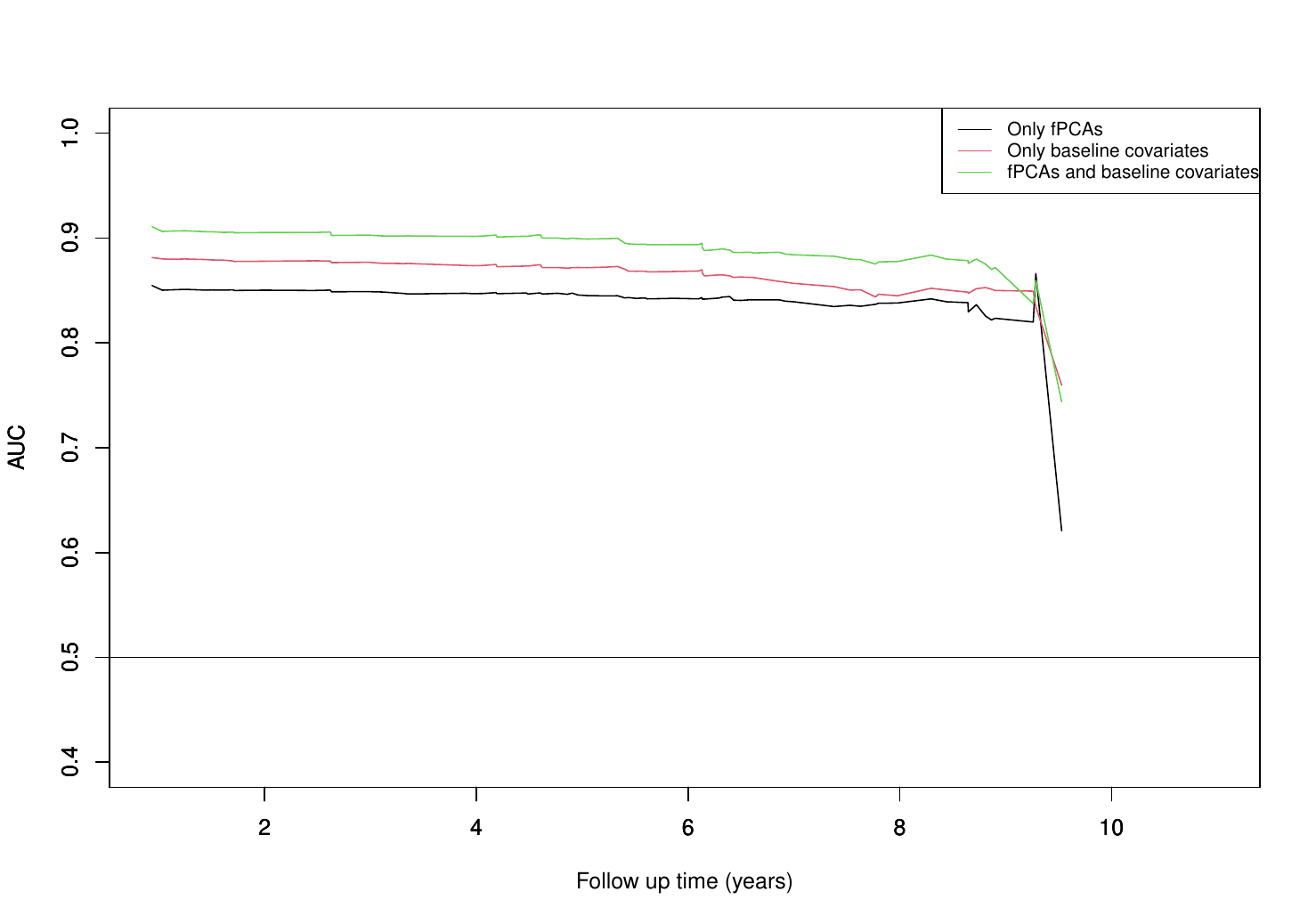}
	    \caption{\footnotesize AUC curves for the three survival models considered: i) only CGM-information; ii) only non-CGM information; iii) CGM and non CGM information.}
	    \label{fig:auc}
         \end{figure}

  In this section, we evaluate and compare the predictive performance of three distinct models through the lens of the time-dependent Receiver Operating Characteristic (ROC) curve:

\begin{enumerate}
\item Employing solely Continuous Glucose Monitoring (CGM) data.
\item Relying exclusively on non-CGM information.
\item Integrating both CGM and non-CGM data.
\end{enumerate}

Figure~\ref{fig:auc} presents the Area Under the Curve (AUC) for each of these models. A notable enhancement in predictive accuracy is observed when CGM data is combined with traditional biomarkers. To determine the statistical significance of the differences observed between the AUC curves, we performed bootstrap resampling with 1000 iterations to calculate the 95\% confidence intervals for the AUC differences. As illustrated in Figure~\ref{fig:difauc2}, the confidence intervals do not encompass 0, underscoring that the comprehensive model incorporating both CGM and non-CGM information significantly outperforms the model based solely on non-CGM data in terms of AUC.
\begin{figure}[ht]
	    \centering
	    \includegraphics[width = 0.6\linewidth,height= 0.6\linewidth]{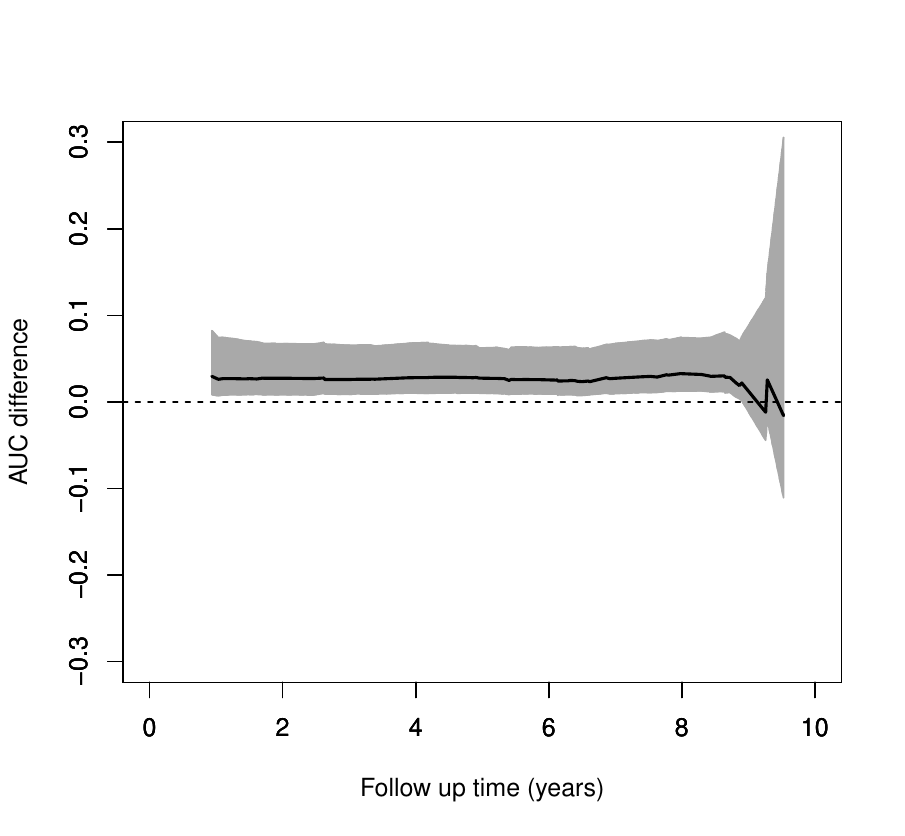}
	    \caption{\footnotesize Difference of AUC over time for the model that contain funcional and non functional CGM information. The 95\% confidence bands calculated by bootstrap resampling with 1000 samples are shown in grey.}
	    \label{fig:difauc2}
         \end{figure}
The Concordance Index (C-index) for the comprehensive model, which includes all imputed CGM data as well as non-CGM information, stands at 0.82, highlighting its superior predictive capacity.

\subsection{Personalized imputations vs. globally imputed CGM  models}\label{sec:personalized}

 \begin{table}[ht!]
            \centering
            \caption{\footnotesize C statistics for the survival models with the different subsets depending on the maximum radius of the confidence bands of the glucodensities, and the number of subjects.
                \label{tab:C}}
             \begin{tabular}{c|c|c}
                Radius & Number of individuals & C index \\
                \hline
                0 & 504 & 0.891 \\
                80 & 511 & 0.891 \\
                90 & 524 & 0.892 \\
                100 & 568 & 0.898 \\
                110 & 566 & 0.899 \\
                120 & 595 & 0.882 \\
                130 & 633 & 0.885 \\
                140 & 686 & 0.881 \\
                150 & 727 & 0.877 \\
                160 & 776 & 0.876 \\
                170 & 822 & 0.865 \\
                180 & 865 & 0.874 \\
                190 & 914 & 0.869 \\
                200 & 992 & 0.875 \\
                370 & 1293 & 0.781 \\ 
            \end{tabular}
        \end{table}

        \begin{figure}[ht!]
	    \centering
	    \includegraphics[width = 0.7\linewidth,height = 0.55\linewidth ]{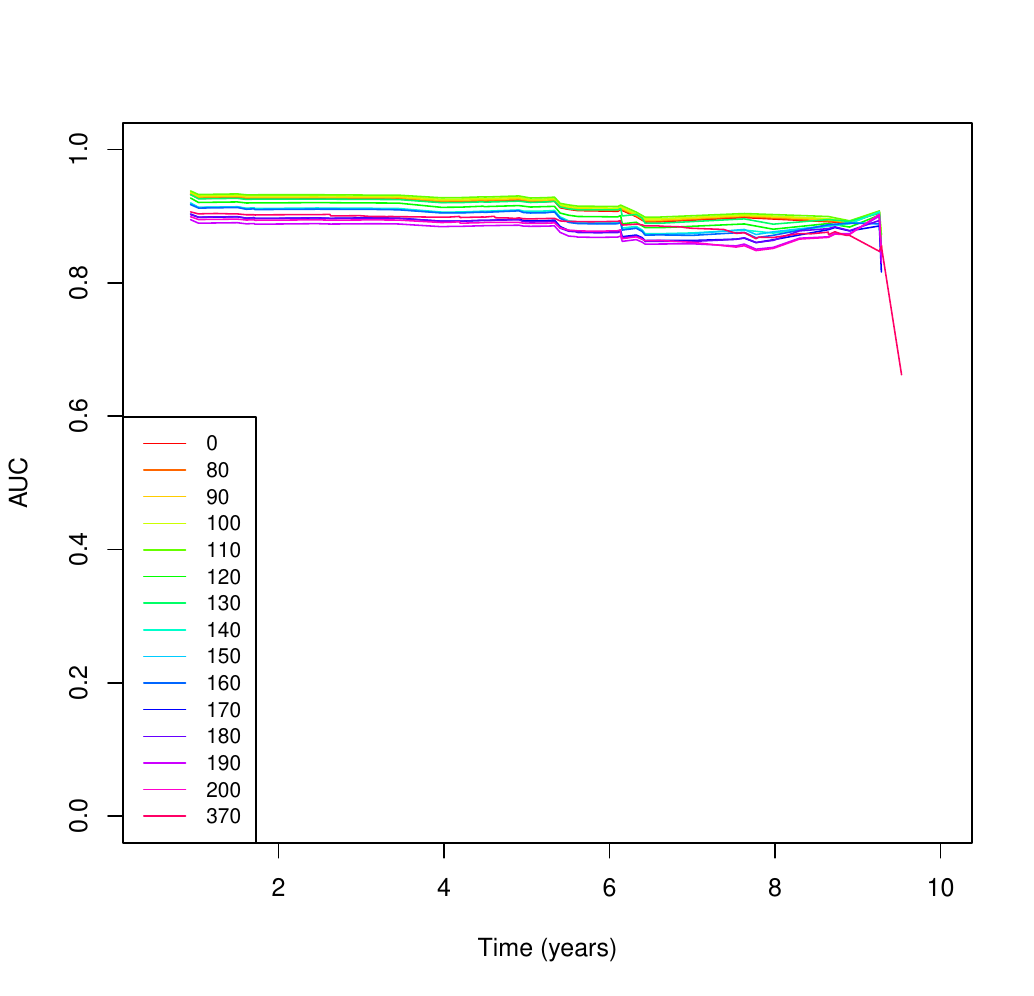}
	    \caption{\footnotesize AUC over time for survival models with  the different subsets depending on the maximum radius of the confidence bands of the quantiles of the glucodensities.}
	    \label{fig:auc_radii}
         \end{figure}

We develop personalized imputation criteria for glucodensity quantiles, contrasting these findings with outcomes derived from predictive models employing a uniform imputation strategy across the entire patient cohort. For a thorough analysis, we examine various radii within a pre-defined grid. Table~\ref{tab:C} displays the Concordance Index (C-index) values obtained for each defined subset, based solely on the imputed glucodensity quantiles. These analyses are conducted under the condition that the maximum confidence interval radius does not surpass the predefined limit, and the table also lists the number of subjects analyzed at each radius. Importantly, our dataset includes 504 individuals with actual glucodensity measurements. Notably, the optimal C-index (0.90), occurs at a radius of 110 and encompasses 62 subjects with imputed Continuous Glucose Monitoring (CGM) data. Figure~\ref{fig:auc_radii} illustrates the Area Under the Curve (AUC) over time for all studied radii, indicating that the overall predictive accuracy exceeds traditional CGM risk assessments, a point elaborated upon in our previous discussions. In overall the C-score with the model with non-functional information is 0.805, there are a improvement of more of ten percent. This improvement highlights the benefit of integrating personalized CGM data into the analysis.

Finally, to elucidate the effectiveness of our personalized imputation approach, Figure~\ref{fig:difauc} displays the conditional Fréchet mean based on the glycemic condition of four representative subjects, showcasing the point-wise results and the predictive bands for the patients, each assigned a distinct maximum radius, and therefore having different levels of uncertainty.
         
\begin{figure}[H]
	    \centering
	    \includegraphics[width = 0.8\linewidth,height= 0.65\linewidth]{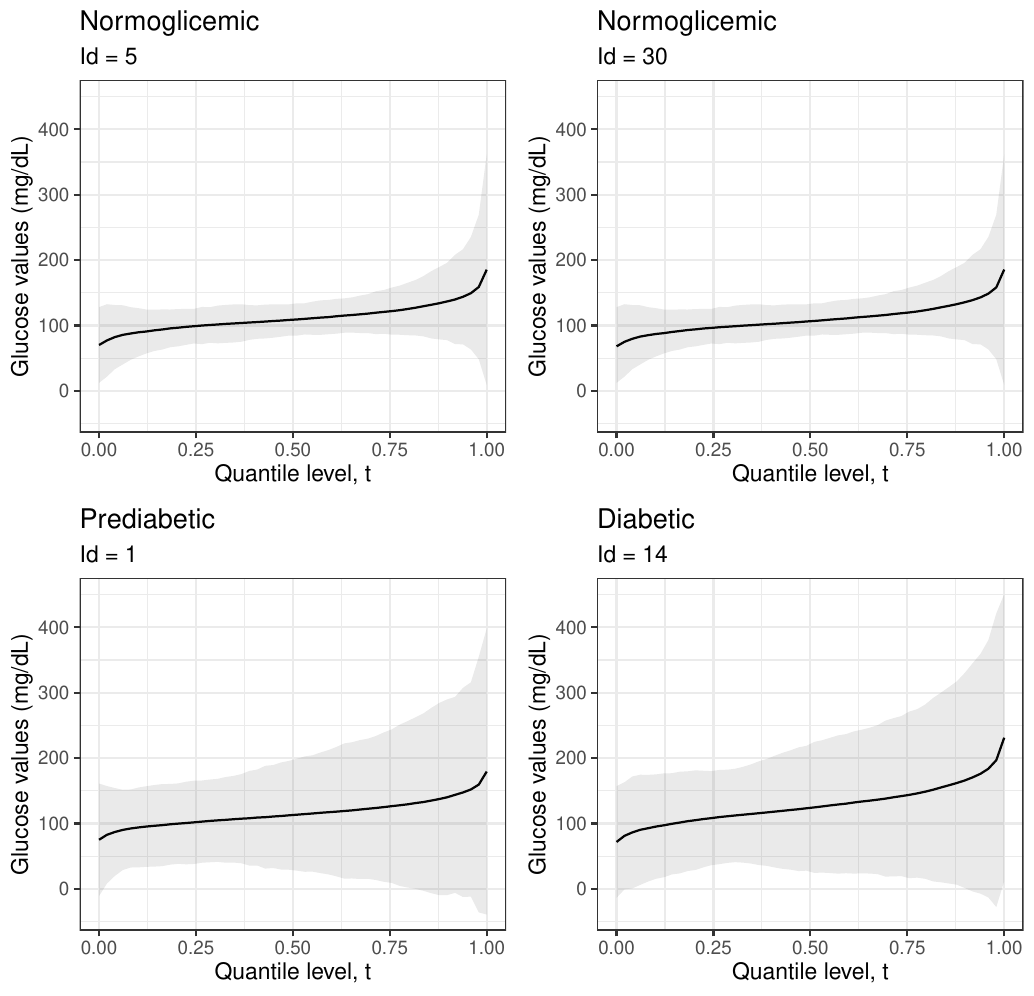}
	    \caption{\footnotesize Conditional Fréchet mean and associated prediction regions in Patients with varying glycemic conditions.}
	    \label{fig:difauc}
         \end{figure}

\section{Discussion}
 \label{sec:disc}

In this study, we delve into a less-explored area of statistical research: personalized imputation strategies for biomedical applications, underscoring the significance of handling missing data, which profoundly impacts both outcomes and predictors \cite{tsiatis2007semiparametric,bertsimas2018predictive}. Unlike previous work focusing primarily on optimal sampling techniques and the development of efficient estimators for scalar variables \cite{lumley2022choosing, https://doi.org/10.1002/sim.9300}, our methodology extends to statistical objects within metric spaces \cite{petersen2019frechet}, marking a significant advancement in statistical modeling for high-resolution medical data.

Methodologically, this paper introduces several notable contributions: a weighted least squares estimator for linear models in metric spaces \cite{petersen2019frechet}, specialized imputation methods for these spaces, innovative non-asymptotic inference techniques for conformal prediction algorithms based on the 2-Wasserstein distance, asymptotic theory for conditional mean imputation within a bounded metric space, and a comprehensive personalized imputation method applicable to various clinical outcomes. This approach underscores the importance of balancing imputation accuracy and reliability, especially in predicting time-to-event outcomes, advocating for a direct assessment of improvements in predictive capacity rather than a sole reliance on biomarkers.

Our approach is demonstrated through its application in biomedicine, specifically in predicting the onset of diabetes in a longitudinal study utilizing data from continuous glucose monitors (CGM). This application not only addresses the widespread use of advanced medical tests in public health initiatives and screening campaigns but also showcases an improvement in model performance by more ten
 percent when integrating CGM data as a digital biomarker , compared to models relying solely on traditional biomarkers. This improvement in predictive accuracy, validated by the C-score, is consistent with traditional biomarkers and findings from other studies \cite{makrilakis2011validation, muhlenbruch2018derivation}.
A key innovation of our model is its ability to integrate high-resolution glucose data through the `glucodensity' concept \cite{matabuena2021glucodensities} using distributional representations \citep{ghosal2021distributional,matabuena2023distributional}, offering novel perspectives for the early identification of diabetes risk. 
Our  results suggest the potential of CGM data to create quantifiable methods to assess the glucose homeostasis of the individual in health-populations, a relatively unexplored topic \cite{hall2018glucotypes}.
 We explore for the first time the incorporation of CGM information and glucodensity into predicting time to diabetes. In the future, it may be useful in establishing diagnostic thresholds for diabetes from a personalized standpoint based on CGM data.

Multiple research directions remain to be explored based on this current research. For longitudinal or multilevel statistical objects, e.g., distributional profiles, the imputation method would need to carefully account for the correlation present within various sub-clusters \citep{goldsmith2015generalized,cui2023fast}. Another interesting direction would be to extend the proposed imputation method to multivariate metric-spaced valued objects, where the distribution of one object could inform another \citep{chiou2014linear,dai2018multivariate}.

By addressing the challenges associated with missing data in digital medicine and the statistical treatment of metric spaces, our study highlights the crucial role of personalization in statistical methodologies, evidenced by a substantial real-world application. As the collection of high-resolution longitudinal data becomes more common, the methodologies introduced herein are poised to become increasingly essential in extensive biomedical studies and the integration of data from wearable devices with genetic information \cite{10.1214/23-AOAS1808}.

\bibliographystyle{unsrtnat}

\end{document}